\documentclass[a4]{article}

\usepackage[a4paper,top=2cm,bottom=2cm,left=3cm,right=3cm,marginparwidth=1.75cm]{geometry}

\usepackage{amsfonts}					
\usepackage{amsmath}					
\usepackage{amssymb}					
\usepackage{amsthm}					
\usepackage[english]{babel}				
\usepackage{caption}					
\usepackage{cite}						
\usepackage{chngcntr}					
\usepackage[shortlabels]{enumitem}		
\usepackage{graphicx}					
\usepackage{epstopdf}					
\usepackage[OT2,T1]{fontenc}				
\usepackage{float}						
\usepackage{framed}					
\usepackage{geometry}					
\usepackage[utf8]{inputenc}				
\usepackage{listings}					
\usepackage{mathtools}					
\usepackage{mathrsfs}					
\usepackage{pgf}						
\usepackage{subcaption}					
\usepackage{tikz}						
\usepackage{tikz-cd}					
\usepackage{underscore}					
\usepackage{xifthen}					
\usepackage{stmaryrd}
\usepackage{complexity}

\usepackage[frozencache, cachedir = .]{minted}					

\usepackage[plainpages=true, colorlinks=true, linkcolor=cyan, citecolor=red!60]{hyperref}		
\usepackage[capitalize, noabbrev]{cleveref}

\usepackage{thmtools}
\usepackage{thm-restate}

\newcommand{\smallbullet}{$\vcenter{\hbox{\tiny$\bullet$}}$}

\usepackage{etoolbox}
\newcommand{\Input}{\ifnumequal{\value{inputlisti}}{1}{\textbf{Input.}\quad}{\phantom{\textbf{Input.}\quad}}}
\newlist{inputlist}{enumerate}{2} 
\setlist[inputlist]{label=\smallbullet, format = \Input, leftmargin = *, align=left, noitemsep, labelindent = *, topsep=1pt}

\newcommand{\Output}{\ifnumequal{\value{outputlisti}}{1}{\textbf{Output.}\quad}{\phantom{\textbf{Output.}\quad}}}
\newlist{outputlist}{enumerate}{2}
\setlist[outputlist]{label=\smallbullet, format = \Output, leftmargin = *, align=left, noitemsep, labelindent = *, topsep=1pt}

\newlist{problemlist}{enumerate}{2}
\setlist[problemlist]{label={}, format = \textbf{Problem.\ }, leftmargin = *, align=left, noitemsep, labelindent = \parindent, topsep=1pt}

\newlist{algosteps}{enumerate}{7}
\setlist[algosteps]{label=\textbf{\arabic*:}, noitemsep, topsep = 1pt, leftmargin=1cm}

\title{\NP-completeness of slope-constrained drawing \\of complete graphs}

\author{C\'{e}dric Pilatte\thanks{Department of Mathematics, 
University of Mons (UMONS), Place du Parc 20, 7000 Mons, Belgium.}
\thanks{Dep.~Math.~and Applications, 
\'{E}cole Normale Sup\'{e}rieure (ENS), Rue d'Ulm 45, 75005 Paris, France.}}

\date{January 2020}

 \newtheoremstyle{underlineNoIndent}
{}        
{}              
{}              
{}    
{}              
{}             
{1.5mm}         
{{\underline{\thmname{#1}\thmnumber{ #2}.}}}

\theoremstyle{plain} 
\newtheorem{theoreme}{Theorem}[section]
\newtheorem{prop}[theoreme]{Proposition}
\newtheorem{lem}[theoreme]{Lemma}

\newtheorem*{theorem*}{Theorem}
\theoremstyle{definition}
\newtheorem{example}[theoreme]{Example}
\newtheorem{defin}[theoreme]{Definition}
\newtheorem{fact}[theoreme]{Fact}
\newtheorem*{question*}{Question}
\newtheorem*{defin*}{Definition}
\newtheorem*{conj*}{Conjecture}
\newtheorem*{lem*}{Lemma}
\theoremstyle{remark}
\newtheorem{remark}[theoreme]{Remark}
\newtheorem*{remark*}{Remark}

\theoremstyle{underlineNoIndent}
\newtheorem*{claim*}{Claim}

\newcommand{\rr}{\mathbb{R}^2}
\DeclareMathOperator{\slopes}{\mathrm{slp}}
\newcommand{\emb}{\hookrightarrow}

\newlang{\SCGD}{SCGD\ }

\makeatletter
\g@addto@macro\normalsize{%
  \setlength\abovedisplayskip{6pt}
  \setlength\belowdisplayskip{6pt}
  \setlength\abovedisplayshortskip{4pt}
  \setlength\belowdisplayshortskip{4pt}
}
\makeatother

\begin{document}
 
\maketitle

\vspace{-0.6cm}

\begin{abstract}
We prove the \NP-completeness of the following problem. Given a set $S$ of $n$ slopes and an 
integer $k\geq 1$, is it possible to draw a complete graph on $k$ vertices in the plane using only 
slopes from $S$? Equivalently, does there exist a set $K$ of $k$ points in general position such that 
the slope of every segment between two points of $K$ is in $S$? We then present a polynomial-time 
algorithm for this question when $n\leq 2k-c$, conditional on a conjecture of R.E.~Jamison. For~$n=k$, 
an algorithm in $\mathcal{O}(n^4)$ was proposed by Wade and Chu. For this case, our algorithm is linear 
and does not rely on Jamison's conjecture.

\end{abstract}

\section{Introduction}

A \emph{straight-line drawing} of an undirected graph $G$ is a representation of $G$ in the plane using 
distinct points for the vertices of $G$ and line segments for the edges. The segments are allowed to 
intersect, but not to overlap, meaning that no segment may pass through a non-incident vertex. The slope of 
a line $l$ is denoted by $\slopes(l)\in \mathbb{R}\cup \{\infty\}$. If $A$ is a set of points, we write 
$\slopes(A)$ for the set of all slopes determined by $A$, i.e. 
$$\slopes(A) = \{\slopes(A_1A_2)\mid A_1, A_2\in A, \, A_1\neq A_2\}.$$ 

The number of slopes used in a straight-line drawing is the number of distinct slopes of the segments in 
the drawing. In 1994, Wade and Chu~\cite{WadeChu} introduced the \emph{slope number} of a graph 
$G$, which is the smallest number $n$ for which there exists a straight-line drawing of $G$ using $n$ 
slopes. This notion has been the subject of extensive research. It was proven independently by Pach, 
P{\'a}lv{\"o}lgyi~\cite{Pach} and Bar{\'a}t, Matou{\v{s}}ek, Wood~\cite{Barat} that graphs of maximum 
degree five may have arbitrarily large slope number. In the opposite direction, Mukkamala and 
P{\'a}lv{\"o}lgyi~\cite{Mukkamala} showed that graphs of maximum degree three have slope number at 
most four, generalizing results in~\cite{Engelstein, Keszegh, Szegedy}. Whether graphs of maximum 
degree four have bounded slope number is still an open problem. Computing the slope number of a 
graph is difficult in general: it is \NP-complete to determine whether a graph has slope number 
two~\cite{Formann}. See also~\cite{Keszegh2, Garg, Bruckner, Hoffmann} for the study of the 
\emph{planar} slope number and related algorithmic questions.

Let us consider the case of the complete graph $K_k$ on $k$ vertices. Let 
$K\subset \mathbb{R}^2$ be the set of points corresponding to the vertices in a straight-line drawing of 
$K_k$. From the definitions, we know that $K$ is in general position and the set of slopes used in the 
drawing is exactly $\slopes(K)$. As in~\cite{JamisonNoncollinear}, we will use the adjective \emph{simple} 
instead of \emph{in general position}, for brevity. It is easily seen that a simple set of $k$ points determines 
at least $k$ slopes~\cite{JamisonNoncollinear, WadeChu}. On the other hand, a straight-line drawing of 
$K_k$ with $k$ 
slopes may be obtained by considering the vertices of a regular $k$-gon. The slope number of $K_k$ is 
thus exactly $k$. 

The slope number of a graph $G$ provides only partial information about the possible sets of slopes of 
straight-line drawings of $G$. Two questions arise naturally:

\begin{enumerate}[itemsep=2pt, topsep=2pt]
	\item What can be said about the straight-line drawings of a graph $G$ that use only a certain 
	number of slopes?
	
	\item Given a set $S$ of slopes, does there exist a straight-line drawing of $G$ using 
	only slopes from~$S$? 
\end{enumerate}

We focus on the case where $G$ is a complete graph. For this case, both questions can be rephrased 
by replacing \emph{straight-line drawings} by \emph{simple sets of points}. The case of complete graphs 
is already difficult and sheds some light on the general situation. As we explain below, the first question is 
still unanswered, in almost all cases, while the second is already \NP-complete when restricted to complete
graphs (\cref{thm:main}). 

Regarding the first question, we have already seen that regular $k$-gons are examples of simple sets of 
$k$ points that use only $k$ slopes. As affine transformations preserve parallelism, the image of a 
regular $k$-gon under an invertible affine transformation is also a simple set of $k$ points with $k$ 
slopes (a set obtained this way is called an \emph{affinely-regular $k$-gon}).\footnote{By 
\emph{regular $n$-gon} or \emph{affinely-regular $n$-gon}, we actually mean the set of vertices of the 
corresponding polygon.} Jamison~\cite{JamisonNoncollinear} proved that these are the only 
possibilities, thereby classifying all straight-line drawings of $K_k$ with exactly $k$ slopes. In the same 
paper, he conjectured a much more general statement. 

\begin{restatable}[Jamison]{conj*}{conjecture} For some constant $c_1$, the following holds. If 
$n\leq 2k-c_1$, every simple set of $k$ points forming (exactly) $n$ slopes is contained in an 
affinely-regular $n$-gon.
\end{restatable}
\noindent
The case $n=k$ corresponds to Jamison's result, and the case $n=k+1$ has been proven 
recently~\cite{Pilatte}. The conjecture is still open for $n = k+2$ and beyond.

\medbreak 
The aim of this paper is to investigate the second question: the algorithmic problem of deciding whether 
a complete graph admits a straight-line drawing that uses only slopes from a given set.

\newcounter{saveCounterTheoreme}
\setcounter{saveCounterTheoreme}{\value{theoreme}}
\newcounter{saveCounterSection}
\setcounter{saveCounterSection}{\value{section}}
\begin{defin}\label{def:scgd}
The \emph{slope-constrained complete graph drawing problem} (\SCGD for short), is the 
following decision problem.
\begin{inputlist} 
	\item A set $S$ of $n$ slopes;
	\item A natural number $k$.
\end{inputlist}

\begin{outputlist}
	\item  \texttt{YES} if there exists a simple set $K$ of $k$ points in the plane such that 
	$\slopes(K) \subseteq S$; 

	\item  \texttt{NO} otherwise.
\end{outputlist}
\end{defin}
 
As a simple set of $k$ points determines at least $k$ slopes, the problem is only interesting when 
$n\geq k$. Wade and Chu~\cite{WadeChu} gave an algorithm with time complexity $\mathcal{O}(n^4)$ 
for the restricted version of the \SCGD problem, where the number of points is equal to the number of 
slopes, i.e.~$k=n$. They asked how to solve the problem when the set of slopes contains more than $k$ 
slopes. 

In this article, we also consider variants of the \SCGD problem where the input is required to satisfy an 
inequality of the type ``$n\leq f(k)$'', for several functions $f$. As we will see, the complexity of the 
restricted \SCGD problem is highly dependent on the choice of $f$, if $\P\neq \NP$. Our results are the 
following (also summarized in \cref{fig:complex}):

\begin{itemize}[itemsep=2pt, topsep=2pt]
	\item \textbf{\cref{sec:npcomp}:} The \SCGD problem is \NP-complete (\cref{thm:main}) 
	in the case where there is no restriction on $n$. A careful examination of the 
	proof shows that the \SCGD problem remains \NP-complete when restricted to 
	$n\leq c k^2$ (for some $c > 0$, which may be chosen to be $2$). The key 
	ingredient is a notion of \emph{slope-generic sets} (\cref{def:slopegeneric}).

	\item \textbf{\cref{sec:algo}:} The \SCGD problem becomes polynomial when the number of slopes 
	is not too large compared to the number of points. More precisely, assuming Jamison's conjecture, 
	there is an $\mathcal{O}(n(n-k+1)^4)$ time algorithm for the \SCGD problem restricted to $n\leq 2k-c_1$, 
	where $c_1$ is the constant appearing in the conjecture. We also give a randomized variant of the 
	algorithm which runs in $\mathcal{O}(n)$ time and gives the correct output with high probability 
	(one-sided Monte-Carlo algorithm). 
\end{itemize}

As mentioned earlier, Jamison's conjecture has been proven for $n=k$ and $n=k+1$. Consequently, our 
algorithm is correct unconditionally when restricted to $n\leq k+1$. Moreover, in this case, it is linear, 
which is easily seen to be optimal. In particular, it improves the $\mathcal{O}(n^4)$ algorithm of Wade 
and Chu~\cite{WadeChu} (which applies to the case $n=k$ only).

\begin{figure}[H]
\centering
\begin{tikzpicture}
\node (0) at (0,0) {};
\node (1) at (1.2,0) {};
\node (2) at (2.4,0) {};
\node (25) at (3.7, 0) {};
\node (3) at (5,0) {};
\node (35) at (6, 0) {};
\node (4) at (7, 0) {};
\node (45) at (8.25,0) {};
\node (10) at (9.5,0) {};

\node (11) at (-1,-0.05) {};

\draw[->] (0.center) -- (10.center);
\fill (0) circle (2pt);
\fill (1) circle (2pt);
\fill (2) circle (2pt);
\fill (3) circle (2pt);
\fill (4) circle (2pt);

\node[above] at (0) {$k$};
\node[above] at (1) {$k+1$};
\node[above] at (2) {$k+2$};
\node[above] at (3) {$2k-c_1$};
\node[above] at (4) {$\,\,c k^2$};

\node[below=2pt] at (0) {$\mathcal{O}(n)$};
\node[below=2pt] at (1) {$\mathcal{O}(n)$};
\node[below=2pt] at (25) {$\mathcal{O}(n(n-k+1)^4)^*$};
\node[below=2pt] at (35) {?};
\node[below=2pt] at (45) {\NP-complete};

\node[above] at (25) {$\ldots$};
\node[above] at (35) {$\,\ldots$};
\node[above] at (45) {$\ldots$};

\node[above] at (10) {$f(k)$};

\node[below] at (8, -0.6) {*assuming Jamison's conjecture.};
\end{tikzpicture}
\caption{Complexity of the \SCGD problem restricted to $n\leq f(k)$ for different functions $f$, where 
$n$ is the number of slopes and $k$ the number of points.}
\label{fig:complex}
\end{figure}
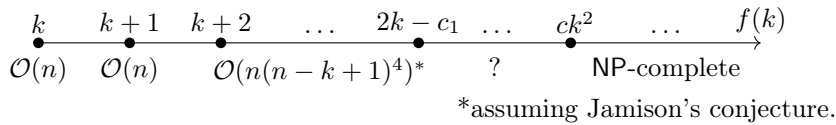

The question of the complexity of the \SCGD problem in the intermediate case, denoted by a question mark in figure 
\cref{fig:complex}, is still unanswered.

\section{Notations and terminology}
Recall that a set of points $A$ is \emph{simple} if no three points of $A$ are collinear. We will say that 
$A$ \emph{has distinct slopes} if it is simple and $\slopes(A_1A_2) \neq \slopes(A_3A_4)$ for every four 
distinct points $A_i$ of $A$. 

We use the term \emph{list} for an ordered sequence. By abuse of notation, if $A=(A_i)_{i\in I}$ is a list, 
we continue to write $A$ for the underlying set $\{A_i\mid i\in I\}$. The notions defined for sets of points 
thus apply to lists of points by ignoring the order structure.

A set of slopes is naturally endowed with a cyclic order, induced from the one on 
$\mathbb{R}\cup \{\infty\}$ (see~\cite[§7.2]{Heyting}). If $S$ is a set of slopes and 
$T=(s_1, \ldots, s_n)$ is a list of slopes, we say the slopes of $T$ are \emph{consecutive slopes} of $S$ 
if $T$ is an interval of the cyclically ordered set $S$. This means that $s_1, s_n$ are the two endpoints 
of the interval and that $s_1, \ldots, s_n$ are all the intermediate slopes, in the correct cyclic order.

Let $\mathrm{Aff}(2, \mathbb{R})$ be the group of invertible affine transformations of the plane. A 
slope can be identified with a point at infinity. There is a natural action of $\mathrm{Aff}(2, \mathbb{R})$ 
on the line at infinity. Therefore, if $\phi \in \mathrm{Aff}(2, \mathbb{R})$, it makes sense to write 
$\phi(s)$ when $s$ is a slope. Let $\mathbb{H}$ denote the subgroup of $\mathrm{Aff}(2, \mathbb{R})$ 
of translations and homotheties. These are precisely the affine transformations that map every line $l$ to 
a line parallel to $l$. In other words, $\mathbb{H}$ is the pointwise stabilizer of the line at infinity (see for 
instance~\cite[Chapter~5]{Meserve}).

\begin{defin}\label{def:emb}
If $A, B$ are two point sets in the plane, we write $A \emb B$ if there is a $\phi \in \mathbb{H}$ with $\phi(A)
\subseteq B$. If $\phi(A) = B$, we write $A \sim B$ and say that $A$ and $B$ are \emph{homothetic}. 
\end{defin}

\begin{remark}
	Note that $A\sim B$ implies $\slopes(A) = \slopes(B)$. The converse is false (see 
	e.g.~\cref{rem:dual}).
\end{remark}

\begin{remark}
	If two simple lists of $n\geq 2$ distinct points $E = (E_1, \ldots, E_n)$ and $F = (F_1, \ldots, F_n)$ 
	satisfy $\slopes(E_iE_j) = \slopes(F_iF_j)$ for all $i\neq j$, there is a unique $\phi\in \mathbb{H}$ 
	such that $\phi(E_i) = F_i$ for all $i$. In particular, $E\sim F$. Note that this assumption is 
	stronger than just $\slopes(E) = \slopes(F)$.
	\label{rem:sim}
\end{remark}

\section{\NP-completeness of the \SCGD problem}

\label{sec:npcomp}

In \cref{sec:motivation}, we informally present some ideas leading to the introduction of the 
notion of \emph{slope-generic sets}. The 
actual definitions are given in \cref{sec:defslopegeneric}, where we show the first properties of 
slope-generic sets. In \cref{sec:proofnpcompl}, we prove that the \SCGD problem is \NP-complete, 
assuming that slope-generic sets can be constructed in polynomial time. This last assertion is proved in 
\cref{sec:prooflemma}.

\subsection{Motivation}

\label{sec:motivation}

A sufficiently general set of $n$ points in the plane determines ${n \choose 2}$ distinct slopes.
Suppose now that we wish to go the other way, and define a 
point set by specifying its set of slopes. First of all, we can only hope to define a point 
set modulo $\mathbb{H}$, since $\slopes(A) = \slopes(B)$ whenever $A$ and $B$ are homothetic.
There is a more fundamental problem: the slopes are not 
independent.\footnote{Three real parameters are needed to specify an element of $\mathbb{H}$. 
Therefore, $2n-3$ real parameters are needed to identify a set of $n$ points modulo $\mathbb{H}$. 
For $n\geq 4$, this is less than the number of slopes, ${n \choose 2}$ (it is equal for 
$n\in \{2, 3\}$).} As soon as we consider four 
distinct points, the six slopes they determine are related by a certain polynomial equation 
(cf. \cref{lem:polynomial}). When we consider a point set $A$ of size $n$, we get a polynomial 
equation for every subset of four points of $A$. Hence, most choices of ${n \choose 2}$ slopes do not 
determine a set of $n$ points modulo $\mathbb{H}$. 

We would like to capture the notion of a point set $A$ having a ``generic'' (i.e.~``ordinary'') set of 
slopes $\slopes(A)$. For many purposes, a set of points can be considered ``generic'' when 
it is simple (i.e.~in general position). In the context of slopes, this is not a sufficiently restrictive 
condition. For example, a regular $n$-gon cannot be considered ``generic'', since it has the very 
special property of forming only $n$ different slopes. 
Having distinct slope is the first (but not the only) condition that we will impose on a point 
set to have a ``generic'' set of slopes.

Our notion should be designed in such a way that any 
``sufficiently random'' point set $A$ should have a ``generic'' set of slopes $\slopes(A)$.
However ``random'' $A$ may be, the elements of $\slopes(A)$ will always satisfy a system 
of polynomial equations given by \cref{lem:polynomial}.
The intuitive idea is that the slopes of a ``generic'' point set $A$ should not satisfy 
more equations of the same form.

%


There is another way to decide whether the slopes of a point set $A$ should be considered ``generic''.
The basic idea is as follows: if $A$ is ``generic'', the knowledge of $\slopes(A)$ should be enough to recover the
original set $A$, modulo $\mathbb{H}$. In other words, for any point set $K$ with slope set $\slopes(K) = \slopes(A)$, 
we should have $K \sim A$. As such, this statement is not entirely correct (we give a refined
heuristic below). Moreover, to ensure that $\slopes(A)$ is really ``generic'', we will need to impose
a similar condition \emph{for many subsets $T$ of $\slopes(A)$}, instead of for $T = \slopes(A)$ only. 
A more elaborate version of the previous idea could be as follows. 

\begin{defin}\label{def:ordinary}
A point set $A$ with distinct slopes is \emph{$k$-ordinary} if every subset 
${T\subseteq \slopes(A)}$ of ${k \choose 2}$~slopes has the following property:
whenever $K$ is a set of $k$ points with $\slopes(K) = T$, 
we have $K \emb A$.
\end{defin}



We cannot hope to say anything interesting about $A$ by considering the subsets of three slopes 
of $\slopes(A)$.\footnote{Given any set $T$ of three distinct slopes, 
there exists a triangle $K$ with $\slopes(K) = T$, and the triangle is unique modulo $\mathbb{H}$. This is 
true whether $T \subseteq \slopes(A)$ or not. Thus, 
each of the ${\lvert\slopes(A)\rvert \choose 3}$ choices of three slopes of $A$ gives a different set $K$, but only 
${|A| \choose 3}$ of those~$K$ will have $K \emb A$. Therefore, no set of more than three 
points is $3$-ordinary.} Hence, from now on, we will only be interested in the case $k \geq 4$. The intuition 
one should keep in mind is that any ``sufficiently random'' set $A$ is $k$-ordinary for all $k > 4$. 
For $k=4$, where the situation is more subtle, \cref{def:ordinary} has to 
be modified to correspond to the behaviour of ``sufficiently random'' point sets. 

%
%
%
%
Let us explain the difference between the cases $k>4$ and $k=4$.
Let $k \geq 4$, let $A$ be a ``sufficiently random'' point set and let $T\subseteq \slopes(A)$ 
be a set of ${k\choose 2}$ slopes. There are two possibilities to consider.
\begin{enumerate}[noitemsep, topsep = 2pt]
	\item The first possibility is that one cannot find a set $K$ of $k$ points with $\slopes(K) = T$.
	For this $T$, the property in \cref{def:ordinary} is vacuously true. This situation is represented in 
	\cref{fig:motivationCase2}. The slopes $s\in T$ must originate from at least $k+1$ points of $A$ 
	(otherwise, they would form the slope set of a subset of $k$ points of $A$, as in \cref{fig:motivationCase1}). 
	For this reason, these slopes are ``independent'', or ``unrelated'', in some appropriate sense.
	
	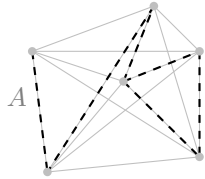
\begin{figure}[H]
	\centering
	\begin{tikzpicture}[scale = 0.8]
		\node[gray] (-1) at (-2, 0.5) {$A$};
		
		\node (0) at (-1.5, -0.75) {};
		\node (1) at (1, -0.5) {};
		\node (2) at (-0.25, 0.75) {};
		\node (3) at (1, 1.25) {};
		\node (4) at (-1.75, 1.25) {};
		\node (5) at (0.25, 2) {};
		
		\draw[gray!50] (0.center) -- (1.center);
		\draw[gray!50] (0.center) -- (2.center);
		\draw[gray!50] (0.center) -- (3.center);
		\draw[gray!50] (0.center) -- (4.center);
		\draw[gray!50] (0.center) -- (5.center);
		\draw[gray!50] (5.center) -- (1.center);
		\draw[gray!50] (5.center) -- (2.center);
		\draw[gray!50] (5.center) -- (3.center);
		\draw[gray!50] (5.center) -- (4.center);
		\draw[gray!50] (2.center) to (3.center);
		\draw[gray!50] (2.center) to (1.center);
		\draw[gray!50] (4.center) to (3.center);
		\draw[gray!50] (3.center) to (1.center);
		\draw[gray!50] (2.center) to (4.center);
		\draw[gray!50] (4.center) to (1.center);
		
		\draw[dashed, thick] (2.center) to (3.center);
		\draw[dashed, thick] (2.center) to (1.center);
		\draw[dashed, thick] (4.center) to (0.center);
		\draw[dashed, thick] (3.center) to (1.center);
		\draw[dashed, thick] (5.center) to (0.center);
		\draw[dashed, thick] (2.center) to (5.center);
		
		\fill [gray!50] (0) circle (2pt);
		\fill [gray!50] (1) circle (2pt);
		\fill [gray!50] (2) circle (2pt);
		\fill [gray!50] (3) circle (2pt);
		\fill [gray!50] (4) circle (2pt);
		\fill [gray!50] (5) circle (2pt);
	\end{tikzpicture}
	\caption{No set of $k$ points in the plane has slope set $T$ ($T$ is in dashed).}
	\label{fig:motivationCase2}
	\end{figure}
	
	\item The other alternative is when we can write $T$ as $T = \slopes(K)$ for some set $K$ of $k$ points.
	Therefore, the slopes in $T$ are ``related''. 
	If $A$ is ``generic'', we expect to have $T = \slopes(B)$ for some subset $B$ of $k$ points of $A$
	(as in the left part of \cref{fig:motivationCase1}, 
	in order to avoid the situation of \cref{fig:motivationCase2}). 
	
	\begin{itemize}[noitemsep, topsep=2pt]
		\item Assume that $k>4$. As $B$ is ``sufficiently random'', we expect $B$ to be the only point
		set modulo $\mathbb{H}$ with slope set $T$. Thus, generically, we have $K \sim B$, so $K \emb A$.
		
		\item Suppose that $k = 4$. Perhaps surprisingly, it turns out that we do not necessarily have $K \sim B$.
		Nevertheless, we can always construct from $K$ a point set $K^*$ (see \cref{def:dual}), 
		for which $\slopes(K^*) = \slopes(K)$ and yet $K^* \not\sim K$ in general. In \cref{fig:motivationCase1}, 
		we have $K \not\sim B$, but $K^* \sim B$.
		This will be true more generally: if $T$ is the set of slopes of a ``generic'' point set $K$,
		we predict every set of four points with slope set $T$ to be homothetic to either $K$ or $K^*$.
		Therefore, we should have either $K \emb A$ (if $K \sim B$) or $K^* \emb A$ (if $K^* \sim B$).
	\end{itemize}
	\begin{figure}[H]
	\centering
	\begin{tikzpicture}[scale = 0.8]
		\node[gray] (-1) at (-2, 0.5) {$A$};
		\node at (0.46, 0.5) {$B$};
		
		\node (0) at (-1.5, -0.75) {};
		\node (1) at (1, -0.5) {};
		\node (2) at (-0.25, 0.75) {};
		\node (3) at (1, 1.25) {};
		\node (4) at (-1.75, 1.25) {};
		\node (5) at (0.25, 2) {};
		
		\draw[gray!50] (0.center) -- (1.center);
		\draw[gray!50] (0.center) -- (2.center);
		\draw[gray!50] (0.center) -- (3.center);
		\draw[gray!50] (0.center) -- (4.center);
		\draw[gray!50] (0.center) -- (5.center);
		\draw[gray!50] (5.center) -- (1.center);
		\draw[gray!50] (5.center) -- (2.center);
		\draw[gray!50] (5.center) -- (3.center);
		\draw[gray!50] (5.center) -- (4.center);
		\draw[gray!50] (2.center) to (3.center);
		\draw[gray!50] (2.center) to (1.center);
		\draw[gray!50] (4.center) to (3.center);
		\draw[gray!50] (3.center) to (1.center);
		\draw[gray!50] (2.center) to (4.center);
		\draw[gray!50] (4.center) to (1.center);
		
		\draw[dashed, thick] (2.center) to (3.center);
		\draw[dashed, thick] (2.center) to (1.center);
		\draw[dashed, thick] (4.center) to (3.center);
		\draw[dashed, thick] (3.center) to (1.center);
		\draw[dashed, thick] (2.center) to (4.center);
		\draw[dashed, thick] (4.center) to (1.center);
		
		\fill [gray!50] (0) circle (2pt);
		\fill [gray!50] (1) circle (2pt);
		\fill [gray!50] (2) circle (2pt);
		\fill [gray!50] (3) circle (2pt);
		\fill [gray!50] (4) circle (2pt);
		\fill [gray!50] (5) circle (2pt);
		
		\begin{scope}[scale = 1.2, xshift = 2cm, yshift  = -0.2cm]
			\node (a0) at (-0.03, 0) {$K$};
			
			\node (a1) at (47/22, -1/22) {};
			\node (a2) at (1, -1/22) {};
			\node (a3) at (-1/4, 3/4) {};
			\node (a4) at (1, -1/2) {};
			
			\draw[dashed, thick] (a1.center) -- (a2.center) -- (a3.center) -- (a4.center) -- cycle;
			\draw[dashed, thick] (a1.center) -- (a3.center);
			\draw[dashed, thick] (a2.center) -- (a4.center);
			
			\fill (a1) circle (1.5pt);
			\fill (a2) circle (1.5pt);
			\fill (a3) circle (1.5pt);
			\fill (a4) circle (1.5pt);
		\end{scope}
		
		\begin{scope}[scale = 0.9, xshift = 4cm, yshift  = 0.8cm]
			\node (a0) at (1.4, 0.5) {$K^*$};
			
			\node (b1) at (1, -0.5) {};
			\node (b2) at (-0.25, 0.75) {};
			\node (b3) at (1, 1.25) {};
			\node (b4) at (-1.75, 1.25) {};
			
			\draw[dashed, thick] (b2.center) to (b3.center);
			\draw[dashed, thick] (b2.center) to (b1.center);
			\draw[dashed, thick] (b4.center) to (b3.center);
			\draw[dashed, thick] (b3.center) to (b1.center);
			\draw[dashed, thick] (b2.center) to (b4.center);
			\draw[dashed, thick] (b4.center) to (b1.center);
			
			\fill (b1) circle (2pt);
			\fill (b2) circle (2pt);
			\fill (b3) circle (2pt);
			\fill (b4) circle (2pt);
		\end{scope}
	\end{tikzpicture}
	\caption{$T$ (in dashed) is determined by a subset $B$ of four points of $A$.}
	\label{fig:motivationCase1}
	\end{figure}
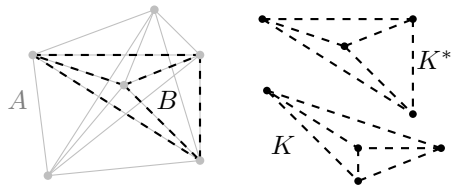
\end{enumerate}

The intuitive reasoning in the case $k=4$ leads to the definition of \emph{slope-generic sets} 
(\cref{def:slopegeneric}). We will see in \cref{lem:embedding} that slope-generic sets 
are automatically $k$-ordinary for all $k \geq 5$.
Slope-generic sets will demonstrate their usefulness in the proof of our main theorem (\cref{thm:main}).
Over the course of \cref{sec:prooflemma}, we will relate this 
purely geometric point of view with 
the previous algebraic considerations (cf. the proof of \cref{lem:example} and \cref{rem:random}).

\subsection{Slope-generic sets}

\label{sec:defslopegeneric}

We start by defining the dual of a list of four points, which is another list of four points that determines 
the exact same slopes, while not being related to the first list by an affine transformation. We denote by 
$\mathfrak{S}_n$ the $n$th symmetric group and by $l(X; PQ)$ the parallel to $PQ$ through $X$.

\begin{defin}
\label{def:dual}
Let $E = (E_1, E_2, E_3, E_4)$ be a simple list of four points. We define a new list $F$ of four points by setting 
$F_1 \coloneqq E_2E_4 \cap l(E_1; E_2E_3)$, $F_2 \coloneqq E_1E_3 \cap l(E_2; E_1E_4)$, 
$F_3 \coloneqq E_2$ and $F_4 \coloneqq E_1$ (see \cref{fig:dual}). The intersections exist as $E$ is 
simple. We call $F$ the \emph{dual} of $E$ and write $F = E^*$.
\end{defin}

\begin{figure}[H]
\centering
\begin{tikzpicture}[scale = 0.33, rotate=0, y=1cm]
\node (1) at (-1.3, -0.28) {};
\fill (1) circle (4pt);
\node (2) at (6.06787, -2.09244) {};
\fill (2) circle (4pt);
\node (3) at (6.68821, 3.98695) {};
\fill (3) circle (4pt);
\node (4) at (4.00005, 0.05809) {};
\fill (4) circle (4pt);

\node (5) at (-0.76032, 5.00887) {};
\fill (5) circle (4pt);
\node (6) at (-6.15248, -2.87198) {};
\fill (6) circle (4pt);


\node [right = 4pt, below = 2pt] at (1) {$E_1,F_4$};
\node [below] at (2) {$E_2,F_3$};
\node [above] at (3) {$E_3$};
\node [below] at (4) {$E_4$};

\node [left] at (5) {$F_1$};
\node [below] at (6) {$F_2$};

\draw (1.center) -- (2.center);
\draw (3.center) -- (6.center);
\draw (5.center) -- (2.center);

\draw [dashed] (2.center) -- (3.center);
\draw [dashed] (1.center) -- (5.center);
\draw [dash dot dot] (1.center) -- (4.center);
\draw [dash dot dot] (2.center) -- (6.center);

\end{tikzpicture}
\caption{Construction of the dual.}
\label{fig:dual}
\end{figure}
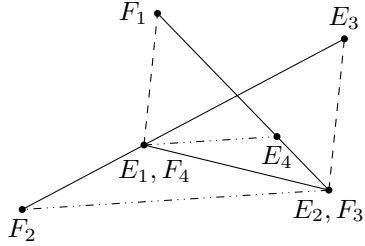

\begin{lem}
Let $E$ be a simple list of four points and let $F = E^*$. Then $\slopes(F) = \slopes(E)$. More precisely, 
$F_{\sigma(1)}F_{\sigma(2)} \parallel E_{\sigma(3)}E_{\sigma(4)}$ for every permutation 
$\sigma \in \mathfrak{S}_4$.
\label{lem:dual}
\end{lem}
\begin{proof}
Of the parallelism conditions that need to be verified, all but one follow directly from the definition. 
For example, $F_1F_4$ and $E_2E_3$ are parallel since, by definition, $F_1$ lies on the parallel to 
$E_2E_3$ through $F_4$. 

The only non-trivial fact is that $F_1F_2 \parallel E_3E_4$. If 
$E_1E_4 \parallel E_2E_3$, we have $F_1 = E_4$, $F_2 = E_3$ and there is nothing to prove. 
Otherwise, this is an application of Pappus's hexagon 
theorem with the collinear triples $(F_1, E_4, E_2=F_3)$ and $(E_3, F_2, E_1=F_4)$.
\end{proof}

\begin{remark}
\label{rem:dual}
\begin{itemize}[noitemsep, topsep=2pt]
	\item If $E$ is a simple list of four points, $\slopes(E^*) = \slopes(E)$, but $E^* \not\sim E$ 
	in general (see \cref{fig:dual,fig:nonexample}). 
	\item It is nonetheless true that $(E^*)^* \sim E$, as can be observed in \cref{fig:nonexample}. 
	To prove it, note that, for every permutation $\sigma \in \mathfrak{S}_4$, we have 
	$E^{**}_{\sigma(1)}E^{**}_{\sigma(2)} 
	\parallel E^{*}_{\sigma(3)}E^{*}_{\sigma(4)} \parallel E_{\sigma(1)}E_{\sigma(2)}$ by \cref{lem:dual}.
	By \cref{rem:sim}, we conclude that $(E^{*})^* \sim E$.
	\item Changing the order of the points of $E$ in \cref{def:dual} does not change the point set $E^*$
	modulo $\mathbb{H}$. In other words, $(E_{\sigma(1)}, \ldots, E_{\sigma(4)})^* \sim (E_1, \ldots, E_4)^*$ 
	for any permutation $\sigma \in \mathfrak{S}_4$.
	This again follows from \cref{lem:dual} and \cref{rem:sim}. 
\end{itemize}
\end{remark}

We now define slope-generic sets. 

\begin{defin}
	A set $A\subset \rr$ is called \emph{slope-generic} if it has distinct slopes and satisfies the 
	following property: for every simple list $E$ of four points, $\slopes(E) \subseteq \slopes(A)$ 
	implies $E\emb A$ or $E^*\emb A$.
	\label{def:slopegeneric}
\end{defin}

\begin{example} 
\begin{itemize}[noitemsep, topsep = 2pt]
	\item The set $A = \{(-2, 2), (-1, 0), (0, 0), (0, 1)\}$ is not slope-generic. Indeed, for 
	$E = \big((1, -1), (-1, 0), (0, 0), (0, 1)\big)$, we have $\slopes(E) = \slopes(A)$. However, $E\not \emb A$ and 
	$E^*\not\emb A$, since $A$ is in convex position, unlike $E$ and $E^*$ (see \cref{fig:nonexample}).
	\item We will give examples of slope-generic sets in \cref{lem:example}. A careful inspection of the proof 
	of \cref{lem:example} reveals that any set of ``sufficiently random'' points is 
	slope-generic (see \cref{rem:random}).
\end{itemize}
\end{example}

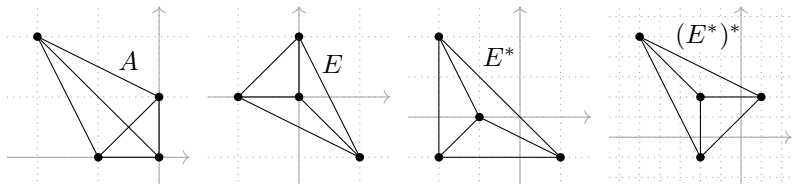
\begin{figure}[H]
\centering

\begin{tikzpicture}[scale = 0.8]
\begin{scope}[xshift = 0.2cm]
\draw[dotted, gray!70] (-2.5, -0.5) grid (0.5, 2.5);
\draw[->, gray!70] (-2.5, 0) -- (0.5, 0);
\draw[->, gray!70] (0, -0.5) -- (0, 2.5);
\node (0) at (-1, 0) {};
\node (1) at (0, 0) {};
\node (2) at (0, 1) {};
\node (3) at (-2, 2) {};
\draw (0.center) -- (1.center) -- (2.center) -- (3.center) -- (0.center);
\draw (0.center) -- (2.center);
\draw (1.center) -- (3.center); 
\node at (-0.5, 1.6) {$A$};
\fill (0) circle (2pt);
\fill (1) circle (2pt);
\fill (2) circle (2pt);
\fill (3) circle (2pt);
\end{scope}

\begin{scope}[xshift = 2.5cm, yshift = 1cm]
\draw[dotted, gray!70] (-1.5, -1.5) grid (1.5, 1.5);
\draw[->, gray!70] (-1.5, 0) -- (1.5, 0);
\draw[->, gray!70] (0, -1.5) -- (0, 1.5);
\node (0) at (-1, 0) {};
\node (1) at (0, 0) {};
\node (2) at (0, 1) {};
\node (3) at (1, -1) {};
\draw (0.center) -- (1.center) -- (2.center) -- (3.center) -- (0.center);
\draw (0.center) -- (2.center);
\draw (1.center) -- (3.center); 
\node at (0.55, 0.55) {$E$};
\fill (0) circle (2pt);
\fill (1) circle (2pt);
\fill (2) circle (2pt);
\fill (3) circle (2pt);
\end{scope}

\begin{scope}[scale = 2/3, xshift = 9.2cm, yshift = 1cm]
\draw[dotted, gray!70] (-2.75, -1.75) grid (1.75, 2.75);
\draw[->, gray!70] (-2.75, 0) -- (1.75, 0);
\draw[->, gray!70] (0, -1.75) -- (0, 2.75);
\node (0) at (-1, 0) {};
\node (1) at (1, -1) {};
\node (2) at (-2, -1) {};
\node (3) at (-2, 2) {};
\draw (0.center) -- (1.center) -- (2.center) -- (3.center) -- (0.center);
\draw (0.center) -- (2.center);
\draw (1.center) -- (3.center); 
\node at (-0.5, 1.5) {$E^*$};
\fill (0) circle (3pt);
\fill (1) circle (3pt);
\fill (2) circle (3pt);
\fill (3) circle (3pt);
\end{scope}

\begin{scope}[scale = 1/3, xshift = 29.3cm, yshift = 1cm]
\draw[dotted, gray!70] (-6.5, -2.5) grid (2.5, 6.5);
\draw[->, gray!70] (-6.5, 0) -- (2.5, 0);
\draw[->, gray!70] (0, -2.5) -- (0, 6.5);
\node (0) at (-2, -1) {};
\node (1) at (1, 2) {};
\node (2) at (-5, 5) {};
\node (3) at (-2, 2) {};
\draw (0.center) -- (1.center) -- (2.center) -- (3.center) -- (0.center);
\draw (0.center) -- (2.center);
\draw (1.center) -- (3.center); 
\fill[white] (-1.6, 5.1) circle (0.9cm);
\fill[white] (-2.2, 5.1) circle (0.9cm);
\fill[white] (-2.6, 5.1) circle (0.9cm);
\fill[white] (-1.3, 5.1) circle (0.9cm);
\node at (-1.6, 5) {$(E^*)^*$};
\fill (0) circle (6pt);
\fill (1) circle (6pt);
\fill (2) circle (6pt);
\fill (3) circle (6pt);
\end{scope}
\end{tikzpicture}

\caption{Examples of point sets with the same set of slopes.}
\label{fig:nonexample}
\end{figure}

\begin{lem}
Let $A$ be slope-generic. Then $A$ has the following property: for every simple set $K\subset \rr$ of at 
least five points, $\slopes(K) \subseteq \slopes(A)$ implies $K\emb A$.
\label{lem:embedding}
\end{lem}

\begin{proof}
Let $K$ be simple with $|K|\geq 5$ and $\slopes(K) \subseteq \slopes(A)$. 

\medbreak

We start by proving that, for every subset $E\subset K$ of four points, we have $E \emb A$. By 
contradiction, there is a quadruple $E=(E_1, E_2, E_3, E_4)$ in $K$ such that $E\not\emb A$. As $A$ is 
slope-generic, $E^*\emb A$. Thus, there exist four points $A_i\in A$ and a map $\phi \in \mathbb{H}$ 
such that $\phi(E^*_i) = A_i$ for $i\in \{1,\ldots,4\}$, where $(E^*_1, E^*_2, E^*_3, E^*_4)=E^*$.

Let $P\in K\setminus E$ and $E' = (P, E_2, E_3, E_4)$. We also have $E'\emb A$ or $E'^*\emb A$, so 
there is a subset $B \subseteq A$ of size $4$ which is homothetic to $E'$ or $E'^*$. Let $\psi \in\mathbb{H}$ 
be the map corresponding to $(E' \text{ or } E'^*) \emb B$. Notice that 
$$\slopes(A_1A_2) \overset{\phi}{=}  \slopes(E^*_1E^*_2) = \slopes(E_3E_4) 
\in \slopes(E') = \slopes(E'^*) \overset{\psi}{=} \slopes(B).$$
As we assumed that $A$ has distinct slopes, this implies that $A_1, A_2\in B$. The same argument with 
$\slopes(A_1A_3)$ and $\slopes(A_1A_4)$ shows that all the $A_i$'s are in $B$, so $B=\{A_1, A_2, A_3, 
A_4\}$.

To summarize, we have 
$(E' \text{ or } E'^*) \overset{\psi}{\sim} B =  \{A_1, \ldots, A_4\} \overset{\phi}{\sim} E^* $.
In particular, we deduce that $\slopes(E') = \slopes(B) = \slopes(E)$. Therefore, 
$$\{\slopes(PE_2), \slopes(PE_3), \slopes(PE_4)\} = \{\slopes(E_1E_2), \slopes(E_1E_3), \slopes(E_1E_4)\},$$
because all the slopes of $B$ are distinct. We can repeat the preceding argument with $(P, E_1, E_3, E_4)$, $(P, E_1, E_2, E_4)$ and $(P, E_1, E_2, E_3)$ in place of $E'$ and deduce the equalities
$$\begin{cases}
\{\slopes(PE_1), \slopes(PE_3), \slopes(PE_4)\} = \{\slopes(E_2E_1), \slopes(E_2E_3), \slopes(E_2E_4)\}\\
\{\slopes(PE_1), \slopes(PE_2), \slopes(PE_4)\} = \{\slopes(E_3E_1), \slopes(E_3E_2), \slopes(E_3E_4)\}\\
\{\slopes(PE_1), \slopes(PE_2), \slopes(PE_3)\} = \{\slopes(E_4E_1), \slopes(E_4E_2), \slopes(E_4E_3)\}.
\end{cases}$$
Taking the union of the left-hand sides and right-hand sides yields $$\{\slopes(PE_1), \slopes(PE_2), 
\slopes(PE_3), \slopes(PE_4)\} = \slopes(E),$$
which is a contradiction since $\lvert\slopes(E)\rvert = \lvert\slopes(B)\rvert = 6$.

\medbreak

For every $E \subset K$ with $|E|=4$, we have proven that there exists a transformation $\phi_E \in 
\mathbb{H}$ such that $\phi_E(E) \subseteq A$. We will now prove that all $\phi_E$ are equal, which 
concludes the proof of the lemma. It is sufficient to prove that $\phi_{E} = \phi_{E'}$ whenever $E, E'$ are 
two subsets of four elements of $K$ with $|E \cap E'| = 3$. 
Let $E\cap E' = \{E_1, E_2, E_3\}$. For $1\leq i< j\leq 3$, we have
$$\slopes(\phi_E(E_i)\phi_E(E_j)) = \slopes(E_iE_j) = \slopes(\phi_{E'}(E_i)\phi_{E'}(E_j)).$$
Since $A$ has distinct slopes, it must be the case that $\{\phi_E(E_i), \phi_E(E_j)\} = \{\phi_{E'}(E_i),\phi_{E'}
(E_j)\}$. It follows that the restrictions $\phi_E|_{E\cap E'}$ and $\phi_{E'}|_{E\cap E'}$ are equal. An affine 
transformation is uniquely determined by its action on three non-collinear points, so $\phi_E = \phi_{E'}$ as 
claimed. 
\end{proof}

\begin{restatable}{lem}{existence}\label{lem:existence}
There exists an algorithm to compute a slope-generic set of size $n$ in time polynomial in $n$. Moreover, 
the coordinates of the constructed points are integers with polynomially many digits, and every slope 
determined by the set is an integer.
\end{restatable}

We postpone the proof of this lemma to \cref{sec:prooflemma} and concentrate 
on the main theorem.

\subsection{Proof of \NP-completeness} 

\label{sec:proofnpcompl}

Let us recall the problem that will be shown to be \NP-complete, as well as some variants which we mentioned in the introduction.

\begin{defin}The \emph{slope-constrained complete graph drawing problem} (\SCGD for short) [\emph{restricted 
to ``$n\leq f(k)$''}], is the following decision problem.

\begin{inputlist} 
	\item A set $S$ of $n$ slopes;
	\item A natural number $k\leq n$ [$\leq f(k)$].
\end{inputlist}

\begin{outputlist} 
	\item \texttt{YES} if there exists a simple set $K$ of $k$ points in the plane such that 
	$\slopes(K)\subseteq S$; 

	\item \texttt{NO} otherwise.
\end{outputlist}

\end{defin}

The condition $k\leq n$ is not restrictive as every simple set of $k$ points determines at least $k$ 
slopes. 

For the model of computation, we will work with the common logarithmic-cost integer 
RAM model~\cite{Boas}. The size of the input $(S,k)$ is measured by the number of bits used 
to represent $k$ and the slopes of $S$ ($k$ may be ignored since $k\leq |S|$). 

This model of computation does not allow the manipulation of arbitrary real numbers.
This is not an issue: to compare the computational complexity of two problems, it is necessary
to use similar models of computation for both. If we were to use, say, the real RAM model, we 
would not be able to talk about \NP-complete problems in the usual sense.

In order to study the complexity
class of the \SCGD problem, we need to specify what inputs are allowed. Since the input slopes $s\in S$ 
must be representable in our model of computation, the two most natural choices might be to only consider 
rational slopes, or integer slopes. We will assume that the slopes are integers. This is not a problem: 
the more we restrict the possible inputs (keeping the same model of computation), the easier 
the problem becomes. Since the \SCGD problem is 
already \NP-complete when restricted to integer slopes, it will still be \NP-hard (and thus \NP-complete, by 
the same proof) for rational slopes.\footnote{Or for any other reasonable choice containing 
the integers, representable within the logarithmic-cost integer RAM model.}

\begin{theoreme}
The \SCGD problem is \NP-complete.
\label{thm:main}
\end{theoreme}

\begin{proof}
We begin by showing that the \SCGD problem is in \NP. Suppose that there exists a simple set 
$K=\{K_1, \ldots, K_k\}$ of $k$ points with $\slopes(K)\subseteq S$.\footnote{We may not directly use 
$K$ as a witness, as the coordinates of the points of $K$ could be arbitrary real numbers. 
Since the slopes are integers, it is true that there always exists another choice of $K$ whose points 
have integer coordinates. However, we would also need to explain that $K$ can 
be chosen to be representable with polynomially many bits. Instead, we choose a more indirect 
witness already containing all the necessary information.}
Let $s_{i, j}\coloneqq \slopes(K_iK_j)
\in S$, for $1\leq i<j\leq k$. Consider the system of linear equations $Y_j-Y_i = s_{i,j}(X_j-X_i)$. A non-trivial 
solution corresponds to an instance of a simple set of $k$ points with slopes contained in $S$. Let the 
witness be the list of triples $(i, j, s_{i, j})$, for $1\leq i<j\leq k$. It is polynomial in the size of the input 
$(S, k)$. If the witness is given, verifying that the corresponding system of linear equations has a non-trivial 
solution is also polynomial in the size of the input. Finally, it is a polynomial time check to verify that 
$s_{i,j}\in S$ and that $s_{i,j} \neq s_{j,k}$ for all distinct $i,j,k$ (the latter condition ensuring that any 
non-trivial solution of the system yields a simple point set).

We now prove that \SCGD is \NP-hard, by showing that \lang{CLIQUE}\footnote{The 
\lang{CLIQUE} decision problem is the following: given a graph $G$ and a positive integer $k$, decide 
whether $G$ contains a clique of size $k$. It is one of the first problems that was shown to be \NP-
complete~\cite{Karp}.} can be polynomially reduced to the \SCGD problem. 

Let $G=(V, E)$ be a finite graph and let $k$ be a positive integer. If $k\leq 4$, solving the clique problem 
with input $(G,k)$ takes polynomial time. Therefore, we only consider the case $k\geq 5$. The idea is to 
consider an embedding of $V$ into a slope-generic set. We construct a slope-generic set $A$ of size $|V|$ in 
polynomial time using the algorithm of \cref{lem:existence}. Fix a bijection $f:V\to A$. Let 
$$S = \{\slopes(f(v)f(w)) \mid vw \text{ is an edge in }G\}\subset \mathbb{Z}.$$
We execute the hypothetical \SCGD algorithm with input $(S, k)$. We claim that the output of the algorithm 
(\texttt{YES} or \texttt{NO}) is exactly the answer to the \lang{CLIQUE} problem with input $(G, k)$. 

\begin{itemize}[noitemsep, topsep=2pt]
	\item If the output is \texttt{NO}, there could not have been a $k$-clique in $G$. Indeed, by 
	contraposition: let $H$ be a $k$-clique in $G$. Then $f(H)$ is a simple set of $k$ points 
	in the plane having slopes in~$S$.

	\item If the output is \texttt{YES}, there is a simple set $K$ in the plane of size $k$ with $\slopes(K) 
	\subseteq S \subseteq \slopes(A)$. As $A$ is slope-generic, $K\emb A$ by \cref{lem:embedding}. 
	Thus, there is no loss of generality in assuming that $K\subseteq A$. 
	The proof of the claim is completed by 
	showing that $f^{-1}(K)$ is a $k$-clique of $G$. Let $A_1, A_2\in K$. We know that $\slopes(A_1A_2)
	\in S$, which means $\slopes(A_1A_2) = \slopes(f(v)f(w))$ for some edge $vw$ in $G$. As $A$ has 
	distinct slopes, we deduce that $\{A_1, A_2\} = \{f(v), f(w)\}$, which implies that $f^{-1}(A_1)f^{-1}(A_2)$ is 
	an edge of $G$. 
\end{itemize}

To conclude, we check that the reduction is polynomial-time (in the size of $G$). 

\begin{itemize}[noitemsep, topsep=2pt]
	\item By \cref{lem:existence}, the set $A$ can be computed in polynomial time (with respect to
	$|V|$). The coordinates of the points of $A$ have polynomially many digits, so the computation time of 
	each $\slopes(f(v)f(w))$ is polynomial in $|V|$.\footnote{By construction (see the proof of 
	\cref{lem:existence} in \cref{sec:prooflemma}), the slope of the line $f(v)f(w)$ is an integer, 
	which is just the sum of the $x$-coordinates of $f(v)$ and $f(w)$.} Thus, the computation of $S$ is 
	polynomial-time in the size of~$G$.
	
	\item The size of the input $(S, k)$ is polynomial in the size of $G$, which concludes the proof. 
	\qedhere
\end{itemize}
\end{proof} 

\begin{remark}
It is well-known that the \lang{HALFCLIQUE}\footnote{The \lang{HALFCLIQUE} problem is the task of 
deciding, given a graph $G$ as input, whether $G$ contains a clique of size $\lceil n/2\rceil$, where $n$ is 
the number of vertices of $G$.} problem is \NP-complete~\cite[Chapter~7]{Sipser}. We can apply the same 
proof wih $k = \lceil |V|/2\rceil$ to get a reduction from \lang{HALFCLIQUE} to the \SCGD problem. With the 
notation from the proof, we have $|S|\leq {|V| \choose 2}\leq 2k^2$. Therefore, the \SCGD problem 
restricted to $n\leq 2k^2$ is also \NP-complete. No attempt has been made here to reduce the constant in 
the inequality.
\end{remark}


\subsection{Construction of slope-generic sets}

\label{sec:prooflemma}

This section is devoted to the proof of \cref{lem:existence}. 
The following lemma gives a  condition for six real numbers to constitute the slope set of a set of four points.

\begin{lem}\label{lem:polynomial}
Let $(m_{i,j})_{1\leq i < j \leq 4}$ be six real numbers. Assume that there exist four distinct points 
$E_1, \ldots, E_4$ in the plane such that $\slopes(E_iE_j) = m_{i,j}$ for all $1\leq i < j \leq 4$. Then 
$$Q(m_{1,2}, m_{1,3}, m_{1,4}, m_{2,3}, m_{2,4}, m_{3, 4}) = 0,$$
where $Q$ is the polynomial 
$$Q(z_1, \ldots, z_6) \coloneqq (z_3-z_5)(z_6-z_2)(z_4-z_1)-(z_2-z_4)(z_5-z_1)(z_6-z_3).$$
\end{lem}

\begin{proof}We can suppose that $E_1 = (0,0)$, translating the four points if necessary. Consider the 
linear system given by the six equations $$y_j - y_i - m_{i, j} (x_j -x_i) = 0, \quad 1\leq i<j\leq 4,$$
where the six unknowns are $x_2, x_3, x_4, y_2, y_3, y_4$ and we fixed $x_1=y_1 = 0$.
It admits the trivial solution where all variables are zero. By assumption, there is another solution, given
by $(x_i, y_i) = E_i$ for $2\leq  i \leq 4$. Hence, the determinant of the system vanishes. This determinant 
computes to $$Q(m_{1,2}, m_{1,3}, m_{1,4}, m_{2,3}, m_{2,4}, m_{3, 4}),$$
concluding the proof of the lemma.
\end{proof}

We will also use the existence of integer sequences with polynomial growth 
avoiding certain additive configurations.

\begin{defin}[Generalized Sidon Sequences] A strictly increasing sequence $\mathcal{C}$ of positive 
integers is a \emph{$B_h$-sequence} if there is no integer $n\geq 1$ which can be expressed as the sum of 
exactly $h$ (non-necessarily distinct) elements of $C$ in two different ways.
\end{defin}

\begin{lem}
	Let $h\geq 2$ be fixed. There is a strictly increasing $B_h$-sequence $(c_i)_{i\in \mathbb{N}}$
	and an algorithm (the ``classic greedy algorithm'') such that
	\begin{enumerate}[noitemsep, topsep=2pt]
		\item the sequence has polynomial growth, more precisely: $c_n = \mathcal{O}(n^{2h-1})$;
		\item the algorithm computes $c_1, \ldots, c_n$ in polynomial time (with respect to 
		$n$).\footnote{The degree of the polynomial depends on $h$.} 
	\end{enumerate}
	\label{lem:sidon}
\end{lem}

For a treatment of a more general case, we refer the reader to~\cite{Cilleruelo} (the proof of \cref{lem:sidon} 
can be found in the introduction). \Cref{lem:example} is the last step before the proof 
of \cref{lem:existence}.

\begin{lem}
\label{lem:example}
The set $A\coloneqq \{(50^{c_i}, 50^{2c_i})\mid i\geq 1\}$ is a slope-generic set for any $B_3$-sequence $
\mathcal{C} = (c_i)_{i\geq 1}$.
\end{lem}

Let us make some comments before starting the proof. In the proof of \cref{thm:main}, we needed 
to be able to exhibit slope-generic sets of arbitrary size. The intuition given in \cref{sec:motivation} 
was that most sets are slope-generic. However, to give an explicit example, one must verify the condition 
in \cref{def:slopegeneric} for a concrete set $A$. This is not an easy task: we have to check the 
condition ($E\emb A$ or $E^* \emb A$) for every possible simple list $E$ of four points 
with $\slopes(E) \subseteq \slopes(A)$. For a finite set $A$, this is a finite computation (considering $E$
modulo $\mathbb{H}$). Here, we give a family of infinite slope-generic 
sets. The points of $A$ are chosen on the parabola $y=x^2$ in order for both the slopes and the points 
to be integers with very simple expressions (the same idea was used in~\cite[Theorem~8.3]{JamisonSurvey}).  

\begin{proof}[Proof of \cref{lem:example}]
Let $A_i = 
(50^{c_i}, 50^{2c_i})$ for $i\in \mathbb{N}$. As $\slopes(A_iA_j) = 50^{c_i}+50^{c_j}$ for $i\neq j$, it is clear 
that $A$ has distinct slopes. 

Suppose that $E\coloneqq\{E_1, \ldots, E_4\}$ is a simple set of four points with $\slopes(E)\subseteq 
\slopes(A)$. We have to show that $E\emb A$ or $E^*\emb A$. We let $m_{i,j} \coloneqq \slopes(E_iE_j) $ 
for $i < j$. As $m_{1,2}\in\slopes(A)$, there are two integers $x_1\neq x_2$ in the sequence $\mathcal{C}$ such 
that $m_{1, 2} = 50^{x_1}+50^{x_2}$. Similarly, $m_{1, 3} = 50^{x_3}+50^{x_4}$, and so on, until $m_{3, 4} = 
50^{x_{11}}+50^{x_{12}}$, for some elements of $\mathcal{C}$ with $x_{2i-1}\neq x_{2i}$ (see \cref{fig:schematicE}). 

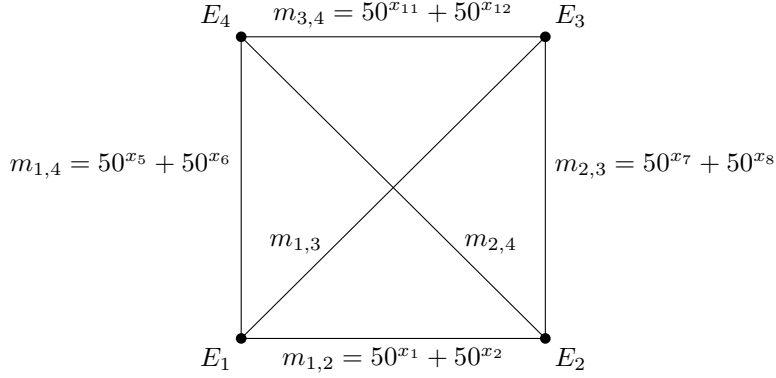
\begin{figure}[H]
\centering
\begin{tikzpicture}[scale = 2]
	\node (0) at (-1, -1) {};
	\node (1) at (1, -1) {};
	\node (2) at (1, 1) {};
	\node (3) at (-1, 1) {};
	
	\fill (0) circle (1pt);
	\fill (1) circle (1pt);
	\fill (2) circle (1pt);
	\fill (3) circle (1pt);
	
	\node[below left] at (0) {$E_1$};
	\node[below right] at (1) {$E_2$};
	\node[above right] at (2) {$E_3$};
	\node[above left] at (3) {$E_4$};
	
	\draw (0.center) -- (1.center) -- (2.center) -- (3.center) -- cycle;
	\draw (0.center) -- (2.center);
	\draw (1.center) -- (3.center);
	
	\node[below] (4) at (0, -1) {$m_{1, 2} = 50^{x_1} + 50^{x_2}$};
	\node[left = 8pt, above] (5) at (-0.5, -0.5) {$m_{1, 3}$};
	\node[above left] (6) at (-1, 0) {$m_{1, 4} =50^{x_5} + 50^{x_6}$};
	\node[above right] (7) at (1, 0) {$m_{2, 3} =  50^{x_7} + 50^{x_8}$};
	\node[right = 8pt, above] (8) at (0.5, -0.5) {$m_{2, 4} $};
	\node[above] (9) at (0, 1) {$m_{3, 4} = 50^{x_{11}} + 50^{x_{12}}$};
\end{tikzpicture}
\caption{Schematic representation of the set $E$ with its six slopes $m_{i,j}$.}
\label{fig:schematicE}
\end{figure}

Since $E$ is simple, two lines $E_iE_{j_1}$ and $E_iE_{j_2}$ passing through the same point $E_i$ cannot 
have the same slope if $j_1\neq j_2$. For example, this tells us that $m_{1, 2} \neq m_{2, 3}$, 
i.e.~$50^{x_1} + 50^{x_2} \neq 50^{x_7} + 50^{x_8}$, which is equivalent to $\{x_1,x_2 \} \neq  \{x_7, x_8\}$.

Putting all the constraints together, we have\footnote{The case $j=7-i$ corresponds to two 
lines $E_iE_j$ and $E_kE_l$ with $\{i,j\}\cap \{k,l\} = \emptyset$, which could be parallel a 
priori. For instance, we do not know (yet) that $m_{1, 2} \neq m_{3, 4}$.}
\begin{equation}
	\begin{cases}
		\text{For }1\leq i\leq 6, x_{2i-1}\neq x_{2i};\\
		\text{For }1\leq i,j\leq 6\text{ and }j\not\in \{i,7-i \},\text{ the sets }\{x_{2i-1}, x_{2i}\}\text{ 
		and } \{x_{2j-1}, x_{2j}\} \text{ are distinct.}
	\end{cases}
\label{eq:distinct}
\end{equation}

By \cref{lem:polynomial} with the slopes $m_{i,j}$, we know that
\begin{equation} Q(50^{x_1}+50^{x_2}, \ldots, 50^{x_{11}}+50^{x_{12}}) = 0.
\label{eq:1}
\end{equation}
Unless many $x_i$'s are actually equal, an equality as \labelcref{eq:1} is unlikely to hold.
The reason is that the polynomial $Q$ (of relatively small degree and coefficients) cannot vanish when 
evaluated at integers of completely different orders of magnitude.

The goal is to use the constraints \labelcref{eq:distinct} and \labelcref{eq:1} to prove that the $x_i$'s 
can take only four distinct values and to know for which indices $i,j$ we have $x_i = x_j$.\footnote{We can 
compare this with \cref{fig:motivationCase2} from \cref{sec:motivation}. If the integers $x_i$ were to 
take more than four different values, it would mean that $\slopes(E)$ (which corresponds 
to $T$ on  \cref{fig:motivationCase2}) is not the slope set of a subset $B$ of four points 
of $A$.} We will see that there are exactly two possibilities, depicted in \cref{fig:solutionE} (here, 
$y_1, \ldots, y_4$ are distinct integers and each $x_i$ is equal to one of the $y_j$'s). 

\begin{figure}[H]
\centering
\begin{tikzpicture}
\begin{scope}[scale = 1.8]
	\node (0) at (-1, -1) {};
	\node (1) at (1, -1) {};
	\node (2) at (1, 1) {};
	\node (3) at (-1, 1) {};
	
	\fill (0) circle (1pt);
	\fill (1) circle (1pt);
	\fill (2) circle (1pt);
	\fill (3) circle (1pt);
	
	\node[below left] at (0) {$E_1$};
	\node[below right] at (1) {$E_2$};
	\node[above right] at (2) {$E_3$};
	\node[above left] at (3) {$E_4$};
	
	\draw (0.center) -- (1.center) -- (2.center) -- (3.center) -- cycle;
	\draw (0.center) -- (2.center);
	\draw (1.center) -- (3.center);
	
	\node[below] (4) at (0, -1) {$50^{y_1} + 50^{y_2}$};
	\node[rotate = 45, left = 8pt, above] (5) at (-0.35, -0.35) {$50^{y_1} + 50^{y_3}$};
	\node[above left] (6) at (-1, 0) {$50^{y_1} + 50^{y_4}$};
	\node[above right] (7) at (1, 0) {$50^{y_2} + 50^{y_3}$};
	\node[rotate = -45, right = 8pt, above] (8) at (0.35, -0.35) {$50^{y_2} + 50^{y_4}$};
	\node[above] (9) at (0, 1) {$50^{y_{3}} + 50^{y_4}$};
\end{scope}

\begin{scope}[scale = 1.8, xshift = 4cm]
	\node (0) at (-1, -1) {};
	\node (1) at (1, -1) {};
	\node (2) at (1, 1) {};
	\node (3) at (-1, 1) {};
	
	\fill (0) circle (1pt);
	\fill (1) circle (1pt);
	\fill (2) circle (1pt);
	\fill (3) circle (1pt);
	
	\node[below left] at (0) {$E_1$};
	\node[below right] at (1) {$E_2$};
	\node[above right] at (2) {$E_3$};
	\node[above left] at (3) {$E_4$};
	
	\draw (0.center) -- (1.center) -- (2.center) -- (3.center) -- cycle;
	\draw (0.center) -- (2.center);
	\draw (1.center) -- (3.center);
	
	\node[below] (4) at (0, -1) {$50^{y_3} + 50^{y_4}$};
	\node[rotate = 45, left = 8pt, above] (5) at (-0.35, -0.35) {$50^{y_2} + 50^{y_4}$};
	\node[below left] (6) at (-1, 0) {$50^{y_2} + 50^{y_3}$};
	\node[below right] (7) at (1, 0) {$50^{y_1} + 50^{y_4}$};
	\node[rotate = -45, right = 8pt, above] (8) at (0.35, -0.35) {$50^{y_1} + 50^{y_3}$};
	\node[above] (9) at (0, 1) {$50^{y_{1}} + 50^{y_2}$};
\end{scope}
\end{tikzpicture}

\caption{The two possibilities for the slopes of $E$, given the constraints \labelcref{eq:distinct} and \labelcref{eq:1}.}
\label{fig:solutionE}
\end{figure}
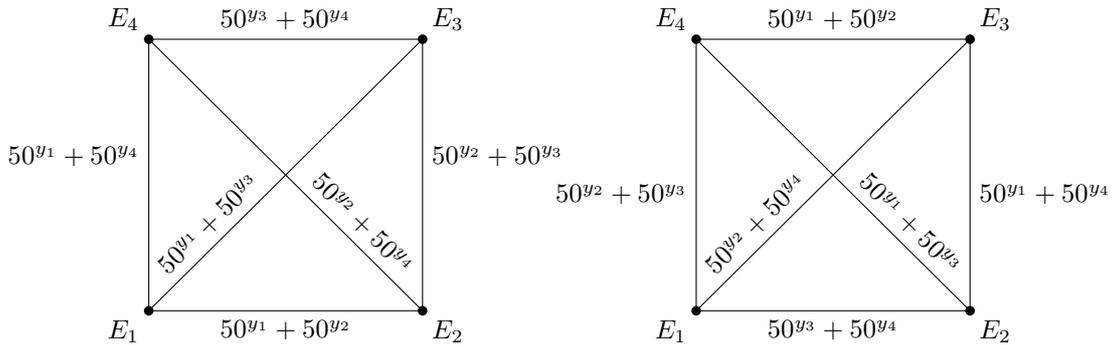

Let $r = |\{x_1, \ldots, x_{12}\}|$ be the number of distinct $x_i$'s. We will show that \labelcref{eq:1} does not 
only hold as an equality between integers, but also holds ``formally'' (or ``symbolically''). To state this 
precisely, we use the language of polynomial rings. Let $\mathbb{Z}[T_1,\ldots,T_r]$ be the 
polynomial ring in $r$ indeterminates. We choose a bijection $\mathcal{J}:\{x_1, \ldots, x_{12}\} \to \{T_1, 
\ldots, T_r\}$.

\begin{claim*}$Q\big(\mathcal{J}(x_1)+\mathcal{J}(x_2), \ldots, 
\mathcal{J}(x_{11})+\mathcal{J}(x_{12})\big)$ is the zero polynomial in $\mathbb{Z}[T_1,\ldots,T_r]$.
\end{claim*}

Essentially, the claim means that, if we replace each $x_i$ by a formal indeterminate, consistently (meaning 
that, if $x_i = x_j$, we replace $x_i$ and $x_j$ by the same indeterminate), \labelcref{eq:1} still holds. To 
be precise, it is $50^{x_i}$ that we replace with some indeterminate in $\{T_1, \ldots, T_r\}$.

\begin{proof}[Proof of claim]
If we expand the left-hand side of \labelcref{eq:1}, simplify, and move half of the terms to the right-hand side, 
we get an equation of the form
\begin{equation}  
	\sum_{i=1}^{48}50^{x_{L_1(i)}+x_{L_2(i)}+x_{L_3(i)}} = \sum_{i=1}^{48}50^{x_{R_1(i)}+x_{R_2(i)}
	+x_{R_3(i)}} 
	\label{eq:2}
\end{equation}
for some known maps $L_k, R_k : \{1, 2, \ldots, 48\} \to \{1, 2, \ldots, 12\}$, $k= 1, 2, 3$ (these maps are just 
obtained by replacing the polynomial $Q$ by its exact expression, given in \cref{lem:polynomial}).
Since the exponents are natural numbers and each sum contains less that fifty terms, equation \labelcref{eq:2}
is equivalent to
\begin{equation}
	\exists \sigma\in \mathfrak{S}_{48},\, \,  \forall  i \in [\![1, 48]\!], \quad x_{L_1(i)}+x_{L_2(i)}+x_{L_3(i)}= 
	x_{R_1(\sigma (i))}+x_{R_2(\sigma (i))}+x_{R_3(\sigma (i))}.
	\label{eq:3}
\end{equation} 
Using the fact that $\{x_1, \ldots, x_{12}\}$ is part of the $B_3$-sequence $\mathcal{C}$, we can 
rewrite \labelcref{eq:3} as
\begin{equation}
	\exists \sigma \in \mathfrak{S}_{48}, \, \, \forall  i \in [\![1, 48]\!], \,\, \exists \tau\in \mathfrak{S}_3, \,\, 
	\forall k \in [\![1, 3]\!], \quad x_{L_k(i)} = x_{R_{\tau(k)}(\sigma(i))}.
	\label{eq:4}
\end{equation} 
	 
To shorten notation, we write $X_i$ instead of $\mathcal{J}(x_i)$.\footnote{Note that $X_i$ is not a
new indeterminate, it is just a notation for the indeterminate $T_j$ which is associated to~$x_i$.} Thus, each 
$X_i$ is an indeterminate and, since $\mathcal{J}$ is a bijection, $X_i = X_j$ if and only if $x_i = x_j$. 
In particular, \labelcref{eq:distinct} becomes
\begin{equation}
	\begin{cases}
		\text{For }1\leq i\leq 6, X_{2i-1}\neq X_{2i};\\
		\text{For all $i,j$ with $j\not\in \{i,7-i \}$, the sets }\{X_{2i-1}, X_{2i}\}\text{ and } \{X_{2j-1}, 
		X_{2j}\} \text{ are distinct.}
	\end{cases}
	\tag{\ref*{eq:distinct}'} \label{eq:distinctX}
\end{equation}
Applying the bijection, \labelcref{eq:4} translates to 
\begin{equation*}
	\exists \sigma \in \mathfrak{S}_{48}, \, \, \forall  i \in [\![1, 48]\!], \,\, \exists \tau\in \mathfrak{S}_3, \,\, 
	\forall k \in [\![1, 3]\!], \quad X_{L_k(i)} = X_{R_{\tau(k)}(\sigma(i))}.
	\tag{\ref*{eq:4}'}
\end{equation*} 
Just like above, this equation can be rewritten as
\begin{equation*}
	\exists \sigma\in \mathfrak{S}_{48},\, \,  \forall  i \in [\![1, 48]\!], \quad X_{L_1(i)}X_{L_2(i)}X_{L_3(i)}= 
	X_{R_1(\sigma (i))}X_{R_2(\sigma (i))}X_{R_3(\sigma (i))},
\end{equation*} 
then as
\begin{equation*}  
	\sum_{i=1}^{48}{X_{L_1(i)}X_{L_2(i)}X_{L_3(i)}} = \sum_{i=1}^{48}{X_{R_1(i)}X_{R_2(i)}X_{R_3(i)}},
\end{equation*}
and finally (by definition of the functions $L_k, R_k$) as
\begin{equation*} 
	Q(X_1+X_2, \ldots, X_{11}+X_{12}) = 0,
	\tag{\ref*{eq:1}'} \label{eq:1X}
\end{equation*}
which proves the claim.
\end{proof}

We will now see what the constraints \labelcref{eq:distinctX} and \labelcref{eq:1X} imply about the $X_i$'s. We will
use the fact that \labelcref{eq:distinctX} and \labelcref{eq:1X} are (in)equalities in the more convenient 
ring $\mathbb{Z}[T_1, \ldots, T_{r}]$. 
By definition of $Q$ (cf. \cref{lem:polynomial}), the equality \labelcref{eq:1X} is equivalent to
\begin{equation} 
	(Z_3-Z_5)(Z_6-Z_2)(Z_4-Z_1)=(Z_2-Z_4)(Z_5-Z_1)(Z_6-Z_3),
	\label{eq:stronger}
\end{equation}
where we used the notations $(Z_1, Z_2, \ldots, Z_6 )\coloneqq (X_{1}+X_{2}, X_{3}+X_{4}, \ldots,  X_{11}+X_{12})$ for 
the equation to fit on a single line. 


Let us look more closely at the factors of \labelcref{eq:stronger}. They are all of the form $Z_i - Z_j$ for some pairs $(i,j)$
with $j\not\in \{i, 7-i\}$. As $Z_i$ is just an abbreviation for $X_{2i-1} + X_{2i}$, we have 
$$Z_i - Z_j = X_{2i-1} + X_{2i} - X_{2j-1} - X_{2j}.$$
Since $j\not\in \{ i, 7-i\}$, we know that $X_{2i-1} \neq X_{2i}$, $X_{2j-1}\neq X_{2j}$ 
and ${\{X_{2i-1}, X_{2i}\}\neq \{X_{2j-1}, X_{2j}\}}$, by~\labelcref{eq:distinctX}. 
Remembering that each of $X_{2i-1}, X_{2i}, X_{2j-1}, X_{2j}$ 
is an element of $\{T_1, \ldots, T_r\}$, we observe that $Z_i - Z_j$ must be a (non-zero) homogeneous 
polynomial with content $1$ and (total) degree $1$.\footnote{The content of a polynomial over $\mathbb{Z}$ 
is the greatest common divisor of its coefficients.} Thus, for these pairs $(i,j)$, the polynomial $Z_i-Z_j $ 
is irreducible, hence prime, in the unique factorization domain $\mathbb{Z}[T_1, \ldots, T_r]$. 
We have just proved that each side of 
\labelcref{eq:stronger} 
is a product of prime elements. By unique factorisation in $\mathbb{Z}[T_1, \ldots, T_r]$,
we can thus identify each factor on one side with one of the factors on the other side (up to a sign).

From the single equation \labelcref{eq:stronger}, we thus deduce three equations (for every choice of signs
and every permutation of the factors of the right-hand side). For example, if we choose all the signs to be $+$ 
and we don't permute the factors, the three equations are
\begin{equation}
\begin{cases}
	Z_3 - Z_5 = Z_2 - Z_4\\
	Z_6 - Z_2 = Z_5 - Z_1\\
	Z_4 - Z_1 = Z_6 - Z_3.
\end{cases}
\label{eq:solvableCaseZ}
\end{equation}
Let us continue this example, and come back to the general case afterwards. After replacing the $Z_i$'s by 
their definition and noticing that the last equation in \labelcref{eq:solvableCaseZ} is redundant, we get
\begin{equation}
	X_1+X_2+X_{11}+X_{12} = X_3+X_4+X_9+X_{10} = X_5+X_6+X_7+X_8.
	\label{eq:solvableCaseX}
\end{equation}
Since every $X_i$ is an indeterminate $T_j$ in $\mathbb{Z}[T_1, \ldots, T_r]$, the only 
way \labelcref{eq:solvableCaseX} can hold is if $X_3, X_4, X_9, X_{10}$ and $X_5, X_6, X_7, X_8$ are 
permutations of $X_1,X_2,X_{11},X_{12}$. Not every pair of permutations is allowed: the 
constraints \labelcref{eq:distinctX} must still be satisfied. Two ways that \labelcref{eq:distinctX} 
and \labelcref{eq:solvableCaseX} can simultaneously be verified are as follows.
\begin{equation}\text{(i)}\,\, \begin{cases}
	X_1=X_3=X_5\\
	X_2=X_9=X_7 \\
	X_{11}=X_4=X_8\\
	X_{12}=X_{10}=X_{6}\\
	X_1, X_2, X_{11}, X_{12} \text{ pairwise distinct}
\end{cases}\quad \text{(ii)}\,\,
\begin{cases}
	X_{11}=X_9=X_{7}\\
	X_{12}=X_3=X_{5} \\
	X_1=X_{10}=X_{6}\\
	X_2=X_4=X_8\\
	X_{11}, X_{12}, X_1, X_2 \text{ pairwise distinct}
\end{cases}
\label{eq:solvableCaseSolved}
\end{equation}

We now return to the general case. We use a computer program to check all the possibilities (the source code
can be found in \cref{appendix:program}). 
With the help of the program, we conclude the following.
\begin{enumerate}[itemsep=2pt, topsep=3pt]
	\item There is only one way to choose three signs and a permutation of the factors that does 
	not lead to a contradiction (it is the choice made in \labelcref{eq:solvableCaseZ}).
	\item There are $2^7$ possibilities in total. Up to the symmetries 
	$X_{2i-1} \leftrightarrow X_{2i}$ (there are $2^6$ combinations of such transpositions), 
	there are actually only two possibilities, given by (\textcolor{cyan}{i}) and (\textcolor{cyan}{ii}) from \labelcref{eq:solvableCaseSolved}.
\end{enumerate}

We can now prove that $A$ is slope-generic. Recall that $X_i = X_j$ if and only if $x_i = x_j$. Without loss of 
generality (performing exchanges $X_{2i-1} \leftrightarrow X_{2i}$ and 
$x_{2i-1} \leftrightarrow x_{2i}$ if necessary), we may thus assume to be in 
one of the following two cases.
$$\text{(i)}\,\, \begin{cases}
	y_1\coloneqq x_1=x_3=x_5\\
	y_2\coloneqq x_2=x_9=x_7 \\
	y_3\coloneqq x_{11}=x_4=x_8\\
	y_4\coloneqq x_{12}=x_{10}=x_{6}\\
	y_1, y_2, y_3, y_4 \text{ pairwise distinct}
\end{cases}\quad \text{(ii)}\,\,
\begin{cases}
	y_1\coloneqq x_{11}=x_9=x_{7}\\
	y_2\coloneqq x_{12}=x_3=x_{5} \\
	y_3\coloneqq x_1=x_{10}=x_{6}\\
	y_4\coloneqq x_2=x_4=x_8\\
	y_1, y_2, y_3, y_4 \text{ pairwise distinct}
\end{cases}$$	

In each case, we define the distinct points $B_i \coloneqq (50^{y_i}, 50^{2y_i})$, for $1\leq i\leq 4$. The 
subset $B=\{B_1, B_2, B_3, B_4\}$ of $A$ is our candidate to verify the slope-genericity of $A$ 
with respect to $E$. 
\begin{itemize}[itemsep=2pt, topsep=4pt]
	\item In the case (\textcolor{cyan}{i}), we have $E\sim B$. To see this, note that the quadruples $(E_1, E_2, E_3, E_4)$ 
	and $(B_1, B_2, B_3, B_4)$ are exactly in the configuration of \cref{rem:sim}. This is an 
	immediate verification: for example, one has $\slopes(E_2E_4) = 50^{x_9}+50^{x_{10}} = 
	50^{y_2}+50^{y_4} = \slopes(B_2B_4)$. Hence, $E\emb A$.
	
	\item In the case (\textcolor{cyan}{ii}), write $E^* = (F_1, \ldots, F_4)$. This time, $(F_1, F_2, F_3, F_4)$ and $(B_1, B_2, 
	B_3, B_4)$ are in the configuration of \cref{rem:sim}, so $E^*\sim B$ and $E^*\emb A$. 	
	\qedhere
\end{itemize}
\end{proof}

\begin{remark} 
	By exploiting symmetry, it is possible to work out all the cases by hand, instead of 
	using a computer program. 
\end{remark}
\begin{remark}\label{rem:random}
	In the previous subsections, we said that one could think of a slope-generic set as any 
	``sufficiently random'' set of points. We now give a more concrete explanation of this 
	intuition, in the light of the proof 
	of \cref{lem:example} (supposing that we
	do not require the slopes to be integers anymore). Let $A$ be a set of points, and let $E$ be a simple set 
	of four points such that $\slopes(E) = \slopes(A)$. Thus, every slope $m_{i,j}$ of $E$ can be written as
	$m_{i,j} = \frac{y_{i,j,2}- y_{i,j,1}}{x_{i,j,2} - x_{i,j,1}}$ for some points $(x_{i,j,k}, y_{i,j,k})$ of $A$.
	By \cref{lem:polynomial}, we have 
	\begin{equation}
	Q\left(\frac{y_{1,1,2}- y_{1,1,1}}{x_{1,1,2} - x_{1,1,1}},  
	\frac{y_{1,2,2}- y_{1,2,1}}{x_{1,2,2} - x_{1,2,1}}, \ldots, \frac{y_{3,4,2}- y_{3, 4,1}}{x_{3, 4,2} - x_{3, 4,1}}\right) = 0.
	\label{eq:intuition}
	\end{equation}
	If $A$ is ``random enough'', equation \labelcref{eq:intuition} cannot hold by an ``arithmetic coincidence''.
	Following the proof of \cref{lem:example}, this means that \labelcref{eq:intuition} still holds true when we
	replace the $x_{i,j,k}$'s and $y_{i,j,k}$'s by formal variables, in a ``consistent'' way. Once we have the formal
	equality, it is conceivable that the same type of arguments as in the second half of the proof 
	can yield the desired result, i.e.~that $E \emb A$ or $E^* \emb A$. Since this is not needed for our main theorem,
	we will not give more details.
\end{remark}

We can now prove \cref{lem:existence}.

\existence*

\begin{proof}
	 First, a $B_3$-sequence $c_1< \ldots< c_n$ is calculated in polynomial time using the greedy 
	 algorithm of \cref{lem:sidon}. Then, one computes $A_i = (50^{c_i}, 50^{2c_i})$ for $1\leq i\leq 
	 n$, and returns $\{A_1, \ldots, A_n\}$. Because $c_i = \mathcal{O}(n^{2\cdot 3-1})$, the computation 
	 of the powers $50^{c_i}$ and $ 50^{2c_i}$ is indeed polynomial in $n$. \qedhere
\end{proof}

\begin{remark}
	The property that 
	 the sequence $(c_i)$ grows polynomially is crucial in the logarithmic-cost model of computation. In 
	 the uniform-cost model, we could just have chosen the $B_3$-sequence $c_i = 4^i$. The number 	
	 $50^{4^n}$ can 
	 be computed in linear time by repeated squaring, even though this number has exponentially many 
	 digits. This is the reason why the uniform-cost RAM model (with multiplication) is not considered to 
	 be a reasonable model of computation (see~\cite[pp. 177-178]{Hromkovic} 
	 and~\cite[\textsection2.2.2]{Boas}).
\end{remark}

\section{Algorithms for the restricted \SCGD problem}

\label{sec:algo}

In this section, we present two polynomial algorithms for the \SCGD problem when the number of slopes 
$n = |S|$ is not much larger than $k$. The first one (\cref{prop:exact}) is deterministic, the 
other one (\cref{prop:probabilistic}) is probabilistic.

\subsection{Affinely-regular polygons}

We give two equivalent definitions and some elementary properties of affinely-regular polygons that can be 
found in~\cite{Coxeter}. For a survey of affinely-regular polygons over an arbitrary field, we refer the reader 
to~\cite{FisherJamison}.

\begin{defin} \label{def:affreg}
Let $n\geq 3$. An affinely-regular $n$-gon is a finite set of points $P$ 
satisfying one of the following equivalent properties:
\begin{itemize}[noitemsep, topsep=2pt]
	\item $P$ is the image of a regular $n$-gon under some $\psi\in \mathrm{Aff}(2, \mathbb{R})$;
	
	\item $P = \{\phi^i(P_0)\mid i\in \mathbb{Z}\}$ for some $\phi\in \mathrm{Aff}(2, \mathbb{R})$ of 
	order $n$ and some $P_0\in \mathbb{R}^2$.
\end{itemize}
\end{defin}

\begin{fact} \label{fact:affreg} 
Let $P$ be an affinely-regular polygon with vertices $P_0, \ldots, P_{n-1}$, in cyclic order (say 
counterclockwise). Let $\phi$ be the unique affine transformation such that $\phi(P_i)=P_{i+1}$ for 
$i=0,1,2$. Then, considering the indices modulo $n$, we have
\begin{enumerate}[noitemsep, topsep=2pt]
	\item $\phi$ has order $n$ and $\phi^i(P_0) = P_i$ for all $i$;
	\item $\slopes(P_{i-k}P_{j+k}) = \slopes(P_iP_j)$ for all $i,j,k$ with $i\neq j$;
	\item If $s_0 \coloneqq \slopes(P_{n-1}P_1)$ and $s_i \coloneqq \slopes(P_0P_i)$ for $1\leq i\leq n-1$, 
	then the slopes of $P$ are precisely $s_0, s_1,\ldots, s_{n-1}$, in this (cyclic) order;
	\item The sequence of boundary slopes $(\slopes(P_0P_1), \slopes(P_1P_2),\ldots,\slopes(P_{n-1}
	P_0))$ is
	\begin{equation*}
		\begin{cases}	(s_1, s_3, s_5, \ldots, s_{n-2}, s_0, s_2, \ldots, s_{n-1})\text{ if $n$ is odd;}\\
					(s_1, s_3, s_5, \ldots, s_{n-1}, s_1, s_3, \ldots, s_{n-1})\text{ if $n$ is even.}
		\end{cases}
	\end{equation*}
\end{enumerate}
\end{fact}

We also recall Jamison's conjecture on affinely-regular polygons.

\conjecture*

\subsection{Model of computation}

We want to give an algorithm for the \SCGD problem when Jamison's conjecture applies, i.e.~when $n\leq 
2k-c_1$. Assuming the conjecture, the sets $K$ that satisfy the \SCGD problem are subsets of 
affinely-regular polygons. However, affinely-regular $n$-gons have irrational slopes as soon as $n\neq 3,4,6$ 
(because $\cos(2\pi/n)$ has degree $\phi(n)/2$ over $\mathbb{Q}$, as was proven by 
D. H. Lehmer~\cite{Lehmer}). The problem is thus trivial if the slopes given as inputs are integers, 
as in \cref{sec:npcomp}.

We will therefore allow the slopes to be arbitrary real numbers,\footnote{An alternative would be to restrict 
to algebraic numbers, as explained in Yap~\cite{Yap}.} and adopt the \emph{real RAM model} described 
in~\cite{Preparata}: the primitive arithmetic operations $+,-,\cdot,/$ and comparisons on real numbers are 
available at unit time cost.

%
%

\subsection{Deterministic algorithm}

We start by solving a problem similar to the restricted \SCGD problem when four points of the set $K$ are 
already given as inputs. 

\begin{lem}
	There is a deterministic algorithm with time complexity $\mathcal{O}(n)$ for the following problem. 
	\label{lem:fourpoints}
\begin{inputlist}
	\item A sorted list $S$ of $n$ slopes; 
	\item A natural number $k\leq n$;
	\item A simple list $(P_0, \ldots, P_3)$ of four points.
\end{inputlist}
\begin{problemlist}
	\item Does there exist a point set $K$ satisfying the following conditions?
\end{problemlist}	
\begin{enumerate}[nosep, wide, labelindent = 54pt, label=\textcolor{black}{(}\roman*\textcolor{black}{)}] 
	 \item $K$ forms an affinely-regular polygon; \label{item:a}
	 \item $P_0,\ldots, P_3$ are consecutive points of $K$ (in that order); \label{item:b}
	 \item $k\leq |K|\leq n$; \label{item:c}
	 \item $\slopes(K)\subseteq S$. \label{item:d}
\end{enumerate}

\end{lem}

\begin{proof}
	Let $\phi$ be the unique affine transformation that maps $P_j$ to $P_{j+1}$, for $j=0, 1, 2$. By 
	\cref{fact:affreg}, if there exists a set $K$ satisfying \labelcref{item:a}, \labelcref{item:b} and \labelcref{item:c}, 
	$\phi$ must have order $|K| \in [k, n]$ (and $K$ is the orbit of $P_0$ under $\phi$). This explains the first 
	steps of the algorithm.
	
	\begin{algosteps}
	 	\item Compute $\phi$ (as a $3\times 3$ matrix). Compute $\phi ^j$, $1\leq j\leq n$.
	 
	 	\item If the order of $\phi$ is in the interval $[k, n]$, call it $d$. Otherwise, return \texttt{NO}. 
	 \end{algosteps}

	Suppose now that $\phi$ has order $d\in [k,n]$. Let $K \coloneqq \{\phi^j(P_0)\mid 0\leq j<d\}$. By 
	\cref{def:affreg} and \cref{fact:affreg}, $K$ is the unique affinely-regular polygon 
	satisfying \labelcref{item:a}, \labelcref{item:b} and \labelcref{item:c}. This means that we only need to check whether 
	$K$ satisfies $\slopes(K)\subseteq S$.  
	
	\begin{algosteps}[resume]
		\item We compute the slopes of $K$. Let $s_0 = \slopes(P_1\phi^{d-1}(P_0))$ and $s_j =  
		\slopes(P_0\phi^j(P_0))$, for $1\leq~j\leq~d-1$. By \cref{fact:affreg}, the slopes of $K$ are 
		exactly $s_0, \ldots, s_{d-1}$, in this order.
		
		\item Return \texttt{YES} if $\slopes(K) \subseteq S$, and \texttt{NO} otherwise. It is possible to 
		check the inclusion in linear time since both sides are sorted. \qedhere
	\end{algosteps}
\end{proof}

Now, we suppose that we already know four consecutive slopes determined by $K$. 

\begin{lem}
	There is a deterministic $\mathcal{O}(n)$ algorithm for the following problem. \label{lem:fourslopes}

\begin{inputlist}
	\item A sorted list $S$ of $n$ slopes; 
	\item A natural number $k\leq n$;
	\item A list $(s_0, \ldots, s_3)$ of four distinct slopes.
\end{inputlist}
				
\begin{problemlist} 
	\item Does there exist a point set $K$ satisfying the following conditions?
\end{problemlist}
\begin{enumerate}[wide, nosep, labelindent = 54pt, label=\textcolor{black}{(}\Roman*\textcolor{black}{)}]
	 \item $K$ forms an affinely-regular polygon; \label{item:i}
	 \item $s_0,\ldots, s_3$ are four \emph{consecutive} slopes of $K$ (in that order); \label{item:ii}
	 \item $k\leq |K|\leq n$; \label{item:iii}
	 \item $\slopes(K)\subseteq S$. \label{item:iv}
\end{enumerate}

\end{lem}

\begin{proof}
We reduce the task to \cref{lem:fourpoints}. We first give the two steps of the reduction and 
then provide the explanations.
\begin{algosteps} 
	\item Compute four distinct points $P_0,\ldots, P_3$ (resp.~$\tilde{P_0}, \ldots, \tilde{P_3}$) 
	satisfying the equalities \labelcref{eq:quad1} (resp.~\labelcref{eq:quad2}) below. Such points exist as 
	$s_0, \ldots, s_3$ are distinct.
	\begin{equation}
		\slopes({P}_0{P}_1) = s_0,\, \slopes({P}_0{P}_2) = s_1,\, \slopes({P}_1{P}_2) = s_2 = 
		\slopes({P}_0{P}_3),\,\text{and}\, \slopes({P}_1{P}_3) = s_3
		\label{eq:quad1}
	\end{equation}
	\vspace{-0.6cm}
	\begin{equation}
		\slopes(\tilde{P}_0\tilde{P}_2) = s_0,\, \slopes(\tilde{P}_1\tilde{P}_2) = s_1 = 
		\slopes(\tilde{P}_0\tilde{P}_3),\, \slopes(\tilde{P}_1\tilde{P}_3) = s_2,\, \text{and}\,
		\slopes(\tilde{P}_2\tilde{P}_3) = s_3
		\label{eq:quad2}
	\end{equation}

	\item For $\overline{P}\in \{(P_0,\ldots, P_3), (\tilde{P_0},\ldots, \tilde{P_3})\}$, run the algorithm 
	of \cref{lem:fourpoints} with $S$, $k$ and $\overline{P}$ as inputs. 
	Return \texttt{YES} if one of the two outputs is \texttt{YES}, and \texttt{NO} if both are 					\texttt{NO}.
\end{algosteps}

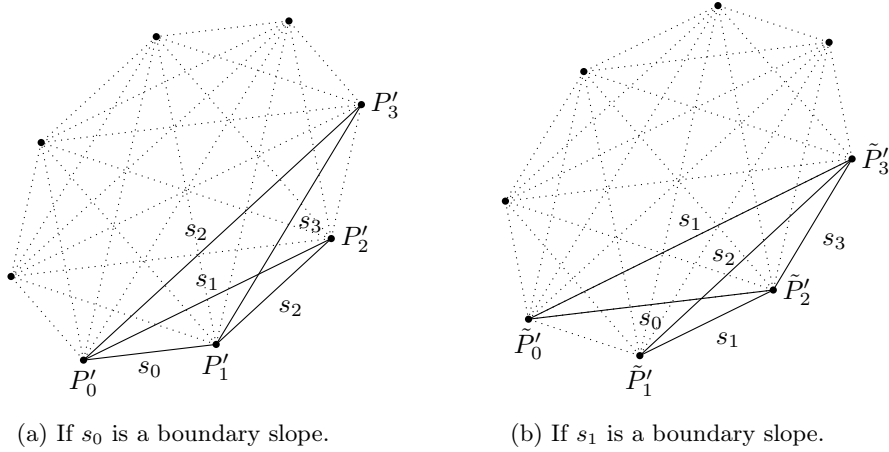
\begin{figure}[H]
\centering
\begin{subfigure}[b]{0.30\linewidth}
\begin{tikzpicture}[scale = 0.45, rotate=0, y=1cm]
\node (1) at (9.99846, -5.99728) {};
\node (2) at (13.87484, -5.53445) {};
\node (3) at (17.23875, -2.41539) {};
\node (4) at (18.11966, 1.53276) {};
\node (5) at (16.00154, 3.99728) {};
\node (6) at (12.12516, 3.53445) {};
\node (7) at (8.76126, 0.41542) {};
\node (8) at (7.88034, -3.53276) {};

\fill (1) circle (3pt); 
\fill (2) circle (3pt); 
\fill (3) circle (3pt); 
\fill (4) circle (3pt); 
\fill (5) circle (3pt); 
\fill (6) circle (3pt); 
\fill (7) circle (3pt); 
\fill (8) circle (3pt); 

\node[below] at (1) {$P_0'$};
\node[below] at (2) {$P_1'$};
\node[right] at (3) {$P_2'$};
\node[right] at (4) {$P_3'$};

\draw (1.center) -- node[below] {$s_0$} (2.center);
\draw (1.center) -- node[above] {$s_1$} (3.center);
\draw (2.center) -- node[below right] {\!$s_2$} (3.center);
\draw (1.center) -- node[left] {$s_2$\,} (4.center);
\draw (2.center) -- node[right] {$s_3$} (4.center);

\draw[dotted] (1.center) -- (5.center);
\draw[dotted] (1.center) -- (6.center);
\draw[dotted] (1.center) -- (7.center);
\draw[dotted] (1.center) -- (8.center);
\draw[dotted] (2.center) -- (5.center);
\draw[dotted] (2.center) -- (6.center);
\draw[dotted] (2.center) -- (7.center);
\draw[dotted] (2.center) -- (8.center);
\draw[dotted] (3.center) -- (4.center);
\draw[dotted] (3.center) -- (5.center);
\draw[dotted] (3.center) -- (6.center);
\draw[dotted] (3.center) -- (7.center);
\draw[dotted] (3.center) -- (8.center);
\draw[dotted] (4.center) -- (5.center);
\draw[dotted] (4.center) -- (6.center);
\draw[dotted] (4.center) -- (7.center);
\draw[dotted] (4.center) -- (8.center);
\draw[dotted] (5.center) -- (6.center);
\draw[dotted] (5.center) -- (7.center);
\draw[dotted] (5.center) -- (8.center);
\draw[dotted] (6.center) -- (7.center);
\draw[dotted] (6.center) -- (8.center);
\draw[dotted] (7.center) -- (8.center);
\end{tikzpicture}
\caption{If $s_0$ is a boundary slope.}
\end{subfigure}\hfil 
\begin{subfigure}[b]{0.30\linewidth}
\begin{tikzpicture}[scale = 0.45, rotate=0, y=1cm]
\node (1) at (11.86349, -6.15679) {};
\node (2) at (15.76403, -4.22321) {};
\node (3) at (18.06681, -0.3592) {};
\node (4) at (17.39444, 3.07606) {};
\node (5) at (14.14701, 4.15806) {};
\node (6) at (10.22686, 2.21476) {};
\node (7) at (7.93557, -1.60339) {};
\node (8) at (8.6113, -5.08273) {};

\fill (1) circle (3pt); 
\fill (2) circle (3pt); 
\fill (3) circle (3pt); 
\fill (4) circle (3pt); 
\fill (5) circle (3pt); 
\fill (6) circle (3pt); 
\fill (7) circle (3pt); 
\fill (8) circle (3pt); 

\node[below] at (8) {$\tilde{P}_0'$};
\node[below] at (1) {$\tilde{P}_1'$};
\node[right] at (2) {$\tilde{P}_2'$};
\node[right] at (3) {$\tilde{P}_3'$};

\draw (8.center) -- node[below] {$s_0$} (2.center);
\draw (1.center) -- node[below right] {$s_1$} (2.center);
\draw (8.center) -- node[above] {$s_1$} (3.center);
\draw (1.center) -- node[left] {$s_2$} (3.center);
\draw (2.center) -- node[below right] {$s_3$} (3.center);

\draw[dotted] (1.center) -- (4.center);
\draw[dotted] (1.center) -- (5.center);
\draw[dotted] (1.center) -- (6.center);
\draw[dotted] (1.center) -- (7.center);
\draw[dotted] (1.center) -- (8.center);
\draw[dotted] (2.center) -- (4.center);
\draw[dotted] (2.center) -- (5.center);
\draw[dotted] (2.center) -- (6.center);
\draw[dotted] (2.center) -- (7.center);
\draw[dotted] (2.center) -- (8.center);
\draw[dotted] (3.center) -- (4.center);
\draw[dotted] (3.center) -- (5.center);
\draw[dotted] (3.center) -- (6.center);
\draw[dotted] (3.center) -- (7.center);
\draw[dotted] (3.center) -- (8.center);
\draw[dotted] (4.center) -- (5.center);
\draw[dotted] (4.center) -- (6.center);
\draw[dotted] (4.center) -- (7.center);
\draw[dotted] (4.center) -- (8.center);
\draw[dotted] (5.center) -- (6.center);
\draw[dotted] (5.center) -- (7.center);
\draw[dotted] (5.center) -- (8.center);
\draw[dotted] (6.center) -- (7.center);
\draw[dotted] (6.center) -- (8.center);
\draw[dotted] (7.center) -- (8.center);
\end{tikzpicture}
\caption{If $s_1$ is a boundary slope.}
\end{subfigure}
\caption{Construction of four consecutive points of $K$ in the proof of \cref{lem:fourslopes}.}
\label{fig:octogons}
\end{figure}

Let us explain why this algorithm is correct.

\begin{itemize}[noitemsep, topsep=2pt]
	\item Suppose that there exists an affinely-regular $d$-gon $K$, of which $s_0, \ldots, s_3$ are 
	consecutive slopes. Then, by \cref{fact:affreg}, at least one of $s_0$ and $s_1$ is a 
	boundary slope of $K$. We have the following alternative (see \cref{fig:octogons}):

	\begin{enumerate}[(a),  noitemsep, topsep=2pt]
		\item  If $s_0$ is a boundary slope of $K$, there are four consecutive vertices ${P}_0', 
		\ldots, {P}_3'$ of $K$ such that 
		\begin{equation}\!\!\!\!\!
			\slopes({P}_0'{P}_1') = s_0,\, \slopes({P}_0'{P}_2') = s_1,\, \slopes({P}_1'{P}_2') = s_2 = 					\slopes({P}_0'{P}_3'),\,\text{and}\, \slopes({P}_1'{P}_3') = s_3.
			\tag{\ref*{eq:quad1}'} \label{eq:quad1tilde}
		\end{equation}

		\item If $s_1$ is a boundary slope of $K$, there are four consecutive vertices $\tilde{P}_0', 	
		\ldots, \tilde{P}_3'$ of $K$ such that 
		\begin{equation}\!\!\!\!\!
			\slopes(\tilde{P}_0'\tilde{P}_2') = s_0,\, \slopes(\tilde{P}_1'\tilde{P}_2') = s_1 = 		
			\slopes(\tilde{P}_0'\tilde{P}_3'),\, \slopes(\tilde{P}_1'\tilde{P}_3') = s_2,\, \text{and}\,
			\slopes(\tilde{P}_2'\tilde{P}_3') = s_3.
			\tag{\ref*{eq:quad2}'} \label{eq:quad2tilde}
		\end{equation}
	\end{enumerate}

	The conditions \labelcref{eq:quad1tilde} uniquely determine the distinct points $\tilde{P}_0,\ldots, \tilde{P}
	_3$, up to an element of $\mathbb{H}$ (the group of homotheties and translations). Therefore, in the 
	case (a), we may assume that $\tilde{P}_0, \ldots, \tilde{P}_3$ are precisely the points constructed in 
	the first step of the algorithm (by applying an element of $\mathbb{H}$ to $K$ if necessary). The 
	same is true with $\tilde{P}_0', \ldots, \tilde{P}_3'$ in the case (b). So we are in the setting of 
	\cref{lem:fourpoints}.
	
	\item Conversely, it is clear that a set $K$ satisfying properties \labelcref{item:a} through \labelcref{item:d} of 		
	\cref{lem:fourpoints} for one of those two quadruples will also satisfy properties \labelcref{item:i} 
	through \labelcref{item:iv}. \qedhere
	
\end{itemize}

\end{proof}

\begin{remark}
	In fact, it is not necessary to consider case (b), because it is possible to prove the following. If $s_0$ 	
	is a slope of an affinely-regular polygon $\tilde{K}$, there is another affinely-regular polygon $K$ 
	which has $s_0$ as a boundary slope and such that $\slopes(\tilde{K}) = \slopes(K)$.
\end{remark}

We can now give the claimed algorithm. The problem here is that we do not know a priori which slopes of 
$S$ will be used by $K$.

\begin{prop}
Assuming Jamison's conjecture, there is an $\mathcal{O}((n-k+1)^4n)$ time deterministic algorithm for the 
\SCGD problem restricted to $n\leq 2k-c_1$:
\label{prop:exact}

\begin{inputlist}
	\item A sorted list $S$ of $n$ slopes 
	\item A natural number $k$ such that $n\leq 2k-c_1$
\end{inputlist}
				
\begin{problemlist} 
	\item Does there exist a simple set $K$ of $k$ points with $\slopes(K)\subseteq S$?
\end{problemlist}

\end{prop}

\begin{proof}
If $n<k$, return \texttt{NO}, as every simple set of $k$ points has at least $k$ slopes.
By Jamison's conjecture, the problem is equivalent to: ``does there exist an affinely-regular $m$-gon $P$ 
with $\slopes(P)\subseteq S$ for some $m\geq k$?''. Let $r\coloneqq n-k$.

\begin{algosteps} 
	\item Let $T$ be the list of the first $r+4$ slopes of $S$. 
\end{algosteps}

Suppose that there exists an affinely-regular $m$-gon $P$ with $\slopes(P)\subseteq S$ for some $m\geq 
k$. We have $\lvert\slopes(P) \cap T \rvert\geq 4$, as otherwise $\lvert\slopes(S)\rvert \geq \lvert\slopes(P) \cup T \rvert \geq  m+
(r+4) - 3 > n$. So, there exist four slopes of $T$ which are also slopes of $P$. 
Let $s_0, \ldots, s_3$ be the first four occurrences of slopes of $P$ in the list $T$, in this order. By construction, the slopes $T$ were chosen to be consecutive slopes of $S$. As $\slopes(P)
\subseteq S$, this implies that $s_0, \ldots , s_3$ must be consecutive slopes of $P$.

\begin{algosteps}[resume] 
	\item For every subsequence of four slopes of $T$, execute the algorithm of 		
	\cref{lem:fourslopes} (with the same $S$ and $k$). If the output is \texttt{YES} for at least one 
	subsequence of four slopes of $T$, return \texttt{YES}. Otherwise, return \texttt{NO}.
\end{algosteps}

As the algorithm in \cref{lem:fourslopes} has runtime $\mathcal{O}(n)$, the time complexity of the 
full algorithm is $\mathcal{O}\left({r+4\choose 4}n\right)$.
\end{proof}

\begin{remark}
\begin{itemize}[noitemsep, topsep=2pt]
	\item If we restrict the inputs to have $n\leq k+M$ for some fixed $M$, we get an algorithm in $
	\mathcal{O}(n)$. This is the optimal complexity since all the slopes have to be taken into account in 
	the worst case.
	
	\item Jamison's conjecture was proven in the cases $n=k$ (by Jamison in his original 
	paper~\cite{JamisonNoncollinear}) and $n=k+1$ (recently, by Pilatte~\cite{Pilatte}). In those cases, the 
	correctness of our algorithm is independent of any assumption.
	
	\item For $n=k$, Wade and Chu presented an algorithm in $\mathcal{O}(n^4)$ for this problem 
	(see~\cite{WadeChu}). We have thus reduced the complexity to $\mathcal{O}(n)$. For $n=k+1$, no 
	polynomial algorithm had been proposed before.
\end{itemize} 
\end{remark}

\subsection{Monte-Carlo algorithm}
	The previous algorithm has runtime $\mathcal{O}((n-k+1)^4n)$, which is $\mathcal{O}(n^5)$ in the 
	worst case. We can improve it to $\mathcal{O}(n)$ if we are willing to use a probabilistic algorithm.

\begin{prop}
	Assuming Jamison's conjecture, there is a one-sided error Monte-Carlo algorithm with running time $
	\mathcal{O}(n)$ for the decision problem described in \cref{prop:exact}.
	\label{prop:probabilistic}
\end{prop}

\begin{proof}
The idea is the same as in the proof of \cref{prop:exact}: we will find quadruples of four 
consecutive slopes in $S$ and apply \cref{lem:fourslopes} with them. 

\begin{algosteps}
	\item Pick one slope $t_1$ of $S$ uniformly at random.

	\item Select the slopes $t_2, \ldots, t_{12}$ of $S$ such that $T = (t_1, \ldots, t_{12})$ are 
	consecutive slopes in $S$.

	\item For each subsequence of four slopes of $T$, use \cref{lem:fourslopes}. If the output is 
	\texttt{YES} for at least one subsequence, return \texttt{YES}. Otherwise, return \texttt{NO}.
\end{algosteps}

\begin{itemize}[noitemsep, topsep=2pt]
	\item If this algorithm outputs \texttt{YES}, the existence of a valid set $K$ is guaranteed by 
	\cref{lem:fourslopes}, so the output is correct. 
	
	\item What is left is to show that the probability of incorrectly outputting \texttt{NO} is bounded away from 
	$1$. Suppose that the correct answer is \texttt{YES}. Equivalently, by Jamison's conjecture, there is an 		
	affinely-regular polygon $P$ of with $\slopes(P)\subseteq S$ and $|P|\geq k$. Let $X$ be the random 
	variable defined by $X = |T\cap \slopes(P)|$ ($S$ and $P$ are fixed and $T$ is random). The algorithm 	
	outputs \texttt{YES} whenever $X\geq 4$. As $\lvert\slopes(P)\rvert \geq |P| \geq k \geq \frac{n+c_1}{2}\geq \frac 12 
	|S|$, we have $$\mathbb{E}(X) = \sum_{1\leq i\leq 12} \mathbb{P}[t_i\in \slopes(P)] = 
	12\cdot \frac{\lvert
	\slopes(P)\rvert}{|S|} \geq 6.$$
	Hence $\mathbb{P}[\text{output is }\texttt{NO}] \leq \mathbb{P}[X<4]= \mathbb{P}[12 - X\geq 9] \leq  
	2/3$, by Markov's inequality, which completes the proof.\qedhere
\end{itemize}
\end{proof}

\section*{Acknowledgements}

I am very grateful to Christian Michaux for his lasting support and valuable advice, and to Pierre 
Vandenhove for the fruitful discussions that we had about slope-generic sets.

\bibliographystyle{plain}
\bibliography{reviewed}

\newpage
\appendix

\section{Source code for \cref{lem:example}}

\label{appendix:program}

This is the source code of the Python program mentioned in the proof of \cref{lem:example}. 
It runs in Python 3.7.7. It uses the Sympy library for symbolic computations: see~\url{https://www.sympy.org}.
It is also available as a Python file on arXiv: \url{https://arxiv.org/src/2001.04671v4/anc}.

\begin{minted}[linenos, fontsize=\footnotesize]{python} 
from time import sleep
import itertools as it
from sympy import * # Symbolic computations

def parse_expression(expression):
  '''
  Input: a linear combination of the X_i's, e.g. 2*X_1-X_3+X_4.
  Output: two lists of terms with coefficients, one with the positive coefficients
      and one with negative coefficients.e.g. [(2, X_1), (1, X_4)], [(-1, X_3)].
  '''
  list_terms = Add.make_args(expression)
  coeff_terms = [Mul.make_args(term) for term in list_terms]
  for i, term in enumerate(coeff_terms):
    if len(term) == 1: # No coefficient -> the coefficient is 1
      coeff_terms[i] = (1, term[0])
  positive_terms = [t for t in coeff_terms if t[0] > 0]
  negative_terms = [t for t in coeff_terms if t[0] < 0]
  return positive_terms, negative_terms

def substitute_all(old_var, new_var, equations, remove_trivial = False):
  '''
  Substitutes the occurrences of old_var by new_var in all equations.
  If 'remove_trivial' is True, all expressions which are zero after substitution are discarded.
  '''
  new_equations = [eq.subs(old_var, new_var) for eq in equations]
  return [new_eq for new_eq in new_equations if new_eq != 0 or not remove_trivial]

def solve_recurs(equations, inequations, substitutions = []):
  '''
  Input: - 'equations' and 'inequations', two lists of expressions (each expression
       is a linear combinations of the X_i's)
       - 'substitutions', a list of pairs of variables (X_i, X_j), indicating that X_i has
         previously been substituted with X_j
  Output: A list of substitutions

  Finds all partitions on the set of variables {X_0, ..., X_11} with the property that:
  - for every expression in 'equations', the reduced expression is zero
  - for every expression in 'inequations', the reduced expression is nonzero.

  By reduced expression, we mean the following. If X_i_1, ..., X_i_k is a set of representatives
  for the partition, and if 'e' is an expression, the reduced version of 'e' is the expression
  obtained by
  1) substituting in 'e' every variable (an element of {X_0, ..., X_11}) by the
     variable X_i_k that is in the same class of the partition;
  2) simplifying the expression as much as possible.

  The partition is represented as a sequence of substitutions. The partition corresponding
  to a list of substitutions is the partition with the fewest number of classes with the
  following property: for every substitution (X_i, X_j), X_i and X_j are in the same class.
  '''
  ans_substitutions = []
  if len(equations) == 0:
    return [substitutions]
  positive_terms, negative_terms = parse_expression(equations[0])
  old_var = positive_terms[0][1] # old_var is in the first equation with a positive coefficient
  for coef, new_var in negative_terms: # old_var must cancel out with some other variable new_var
                                 # that appears in the first equation with a negative coefficient
    new_inequations = substitute_all(old_var, new_var, inequations)
    for nonzero in new_inequations:
      if nonzero == 0: # The substitution old_var <- new_var leads to a contradiction
        break
    else: # Perform the substitution and make a recursive call
      new_equations = substitute_all(old_var, new_var, equations, remove_trivial = True)
      sub = solve_recurs(new_equations, new_inequations, substitutions+[(old_var, new_var)])
      if sub is not None:
        ans_substitutions.extend(sub)
  return ans_substitutions

def pretty_print_sub(substitutions):
  '''Prints the given list of substitutions as a sequence of equalities.'''
  partition = []
  sub_to_str = lambda t: (str(t[0]).ljust(4), str(t[1]).ljust(4))
  str_substitutions = map(sub_to_str, substitutions)
  for old, new in str_substitutions:
    for partition_class in partition: # We search for the class of old and new in the partition
      if old in partition_class or new in partition_class:
        partition_class.update({old, new})
        break
    else: # Create a new class in the partition with old and new
      partition.append({old, new})
  for partition_class in partition:
    print(" = ".join(partition_class))

if __name__ == "__main__":
  # Define the variables X_i, the Z_i's and the factors appearing in the proof
  X = [None] + list(symbols('X_1:13')) # We do not use X[0] to match the notations of the proof
  Z = [None] + [X[2*i-1] + X[2*i] for i in range(1, 7)]
  factors_LHS = [Z[3]-Z[5], Z[6]-Z[2], Z[4]-Z[1]]
  factors_RHS = [Z[2]-Z[4], Z[5]-Z[1], Z[6]-Z[3]]

  nonzero_expressions = [] # List of conditions of the form 'expr != 0' satisfied by the X_i's
  for i in range(1, 7):
    nonzero_expressions.append(X[2*i-1] - X[2*i])
    nonzero_expressions.extend([Z[i] - Z[j] for j in range(1, i) if i + j != 7])

  cnt_print = 0
  max_num_sol_to_print = 10
  all_solutions = []
  for signs in it.product([+1, -1], repeat = 3): # Choose 3 signs
    if prod(signs) == 1:
      for permuted_factors_RHS in it.permutations(factors_RHS):
        # Each factor on the LHS must equal a factor on the RHS, up to a sign
        equations = [factors_LHS[i] - signs[i]*permuted_factors_RHS[i] for i in range(3)]

        # Printing parameters
        cnt_print += 1
        print("Case", cnt_print, "out of 24.")
        print("-> The signs are", ", ".join([str(sign).rjust(2, "+") for sign in signs]))
        print("-> Permutation", (cnt_print-1)%6+1, "out of 6.")
        print("The system of equations is:")
        for eq in equations:
          print(str(eq)+" = 0")
        print("Computing the solutions...")

        answer = solve_recurs(equations, nonzero_expressions)
        all_solutions.extend(answer)

        # Printing solutions
        print("...there are", len(answer), "solutions.")
        if len(answer) > 0:
          num_sol_to_print = min(len(answer), max_num_sol_to_print)
          print("\nFor example, here are", num_sol_to_print, "solutions:\n")
          for i in range(num_sol_to_print):
            pretty_print_sub(answer[i])
            print()
          print("There are", len(answer)-num_sol_to_print, "more solutions.\n")
        print()
        sleep(1)
\end{minted}

\end{document}